\documentclass[11pt]{article}
\usepackage{amsmath,amssymb,latexsym,cite}
\usepackage{pstricks,graphicx}
\usepackage{subfigure,amsthm,fullpage}

\def\tr{\mathop{\rm tr}\nolimits}       %

\newtheorem{theorem}{Theorem}
\newtheorem{corollary}{Corollary}
\newtheorem{lemma}{Lemma}

\def\xcostb{\Gamma}

\def\scostb{\Lambda}

\def\jpwr{\scostb}
\def\xpwr{\xcostb}

\def\stone{\sigma_t}
\def\st{\sigma_t^2}

\def\corr{\rho}
\def\boost{\alpha}
\def\optboost{\alpha_{0}}
\def\Rbin{R_{\mathrm{bin}}}
\def\rateeps{\epsilon_1}
\def\whpeps{\epsilon_2}

\def\prob{\mathbb{P}}
\def\expe{\mathbb{E}}
\def\mbf{\mathbf}
\def\mc{\mathcal}
\def\argmin{\mathop{\rm argmin}}
\def\argmax{\mathop{\rm argmax}}

\def\diag{\mathop{\rm diag}}

\newcommand{\ip}[2]{\left\langle #1,\ #2 \right\rangle}
\newcommand{\norm}[1]{\left\| #1 \right\|}

\newcommand{\mi}[2]{\ensuremath{I \left( #1 \ \wedge \ #2 \right)}}

\def\scostb{\Lambda}

\def\denc{\phi}
\def\ddec{\psi}
\def\renc{\Phi}
\def\rdec{\Psi}

\def\stdmaxerr{\varepsilon}

\def\avgerr{\bar{\varepsilon}}

\def\capeps{\epsilon_c}

\def\sw{\sigma_w^2}

\def\Sx{\Sigma_X}
\def\Sw{\Sigma_W}

\def\st{\sigma_t^2}

\newcommand{\noi}[1]{\sigma_{#1}^2}

\title{Relaxing the Gaussian AVC
\thanks{The work of A.D. Sarwate and M. Gastpar was supported in part by the National Science Foundation under award CCF-0347298.  A.D. Sarwate was also supported by the California Institute for Telecommunications and Information Technology (CALIT2) at UC San Diego.  Some of these results were presented at ISIT 2006 \cite{SarwateG:06gavc}, Allerton 2006 \cite{SarwateG:06brother}, CISS 2008 \cite{Sarwate08:ciss}, ISIT 2008 \cite{SarwateG:08isit} and appear in the first author's dissertation \cite{Sarwate:08thesis}.}}

\author{Anand~D.~Sarwate\thanks{A.D. Sarwate is with the Toyota Technological Institute at Chicago, 6045 S. Kenwood Ave., Chicago, IL 60637 USA (e-mail: \texttt{asarwate@ttic.edu}).}  \and~Michael~Gastpar \thanks{ M. Gastpar is with the the Department of Electrical Engineering and Computer Sciences, University of California, Berkeley, CA 94720 USA, and with the School of Computer and Communication Sciences, Ecole Polytechnique Fdrale (EPFL), 1015 Lausanne, Switzerland (e-mail: 
\texttt{michael.gastpar@epfl.ch}).}}

\date{\today}

\begin{document}

\maketitle

\begin{abstract}
The arbitrarily varying channel (AVC) is a conservative way of modeling
an unknown interference, and the corresponding capacity results
are pessimistic.  We reconsider the Gaussian AVC by relaxing
the classical model and thereby weakening the adversarial nature of
the interference.  We examine three different relaxations. First, we
show how a very small amount of common
randomness between transmitter and receiver is sufficient to achieve
the rates of fully randomized codes. Second,
akin to the dirty paper coding problem, we study the impact of
an additional interference known to the transmitter.  We provide 
partial capacity results that differ significantly from the standard AVC.
Third, we revisit a Gaussian MIMO AVC in which the interference is
arbitrary but of limited dimension. 
\end{abstract}

\section{Introduction}

The arbitrarily varying channel is an information-theoretic model of communication under worst-case noise \cite{BlackwellBT:60random,LapidothN:98survey}.  In the Gaussian AVC (GAVC) \cite{HughesN:87gavc} an additive white Gaussian noise (AWGN) channel is modified by adding a power-constrained jamming interference signal.  As in the discrete AVC with constraints \cite{CsiszarN:88constraints}, the capacity is well defined when the power constraints $\xcostb$ and $\scostb$ on the input and jammer are required to hold almost surely.  When the encoder and decoder share common randomness, the jammer can be no more harmful than Gaussian noise, but without common randomness the capacity is zero when $\xcostb \le \scostb$ because the jammer can simulate the encoder and ``symmetrize'' the channel.  If $\xcostb > \scostb$ then under average error the jammer is again no worse than Gaussian noise.  

The GAVC gives one way of understanding the impact of uncertainty on the capacity of point-to-point channels.  However, the dichotomy for deterministic coding reflects the effect of worst-case analysis.  By contrast, by assuming the interference comes from power-limited but arbitrary random noise, it can be shown that the worst-case noise is Gaussian and the capacity is the AWGN capacity \cite{Lapidoth:96additive,Diggavi98thesis}.  Similarly, by allowing feedback and causal coding, arbitrary interference in an ``individual channel'' model may also look Gaussian \cite{LomnitzF:11ind}.  In this paper we reexamine the GAVC model to see how different variants of the model shed insights into what can be achieved against ``worst-case'' interference.

We describe three variants of the GAVC model:
	\begin{enumerate}
	\item In the first model we allow the encoder and decoder to share a limited amount of common randomness.  In particular, we show that $O(\log n)$ bits of common randomness are sufficient to achieve the randomized coding capacity of the AVC, where $n$ is the coding blocklength.  Essentially, a small amount of randomness is sufficient to make the most malicious interference as harmless as random noise.
	\item In the second model, in addition to the noise and the jammer, there
is an additional interference which is known only to the transmitter.
A new achievable rate is found for this setting.  Perhaps
somewhat surprisingly, the presence of this additional
interference increases the capacity, helping to beat the jammer.
Capacity results are found for special cases.
	\item The third model is the Gaussian MIMO AVC under fully randomized coding \cite{HughesN:88vgavc}.  In addition to a power constraint, we also constrain the dimensionality of the jamming signal.  In general this leads to higher rates, and we find the exact capacity for the case of 2 transmit and 2 receive antennas under interference from a single-antenna jammer.
	\end{enumerate}

More generally, these relaxations of the Gaussian AVC shed some light on the nature of worst-case interference.  There are two ways in which the jammer behaves in a worst-case manner.  The capacity dichotomy for deterministic coding arises because the jammer can simulate the valid encoder.  When the capacity is positive, it is limited by the jammer choosing the worst noise distribution.  Our work shows that if the jammer cannot implement these strategies, the corresponding capacity is often higher.  Because these three models are somewhat different from each other, we introduce the relevant definitions with the results\footnote{\textbf{Note to reviewers}:  For ease of reviewing, the proofs are provided in the relevant sections.  However, it is our intention to move most of the technical lemmas to the appendix so that the main body of the paper can be read easily.}.

\section{Limited Common Randomness}

Our first relaxation of the GAVC is a classic one -- we allow the encoder and decoder to share common randomness that is unknown to the jammer \cite{HughesN:87gavc}.  However, in contrast to previous works, we focus on the amount of randomization, or key size.   Our main result is that as in the
discrete case, we can use a sub-exponential number (in the blocklength $n$) of codebooks to obtain an asymptotically decreasing upper bound on the probability of error.

\begin{figure}
\begin{center}
\psset{xunit=0.06cm,yunit=0.06cm,runit=0.06cm}
\begin{pspicture}(0,5)(130,50)
\rput(5,30){$i$}
\psline{->}(10,30)(20,30)
\psframe(20,25)(40,35)
\rput(30,30){\textsc{Enc}}
\psline{->}(30,15)(30,25)
\rput(30,10){$(\renc,\rdec)$}
\psline{->}(40,30)(90,30)
\rput(46,35){$\mbf{X}_i$}

\pscircle(55,30){4}
\pscircle(70,30){4}
\psline(55,26)(55,34)
\psline(70,26)(70,34)
\psline{->}(55,40)(55,34)
\psline{->}(70,40)(70,34)
\rput(55,45){$\mbf{s}$}
\rput(70,45){$\mbf{W}$}

\rput(83,35){$\mbf{Y}$}
\psframe(90,25)(110,35)
\rput(100,30){\textsc{Dec}}
\psline{->}(100,15)(100,25)
\rput(100,10){$(\renc,\rdec)$}
\psline{->}(110,30)(120,30)
\rput(125,30){$\hat{i}$}

\end{pspicture}
\caption{The Gaussian arbitrarily varying channel under randomized coding.
  \label{fig:randgavc}}
\end{center}
\end{figure}
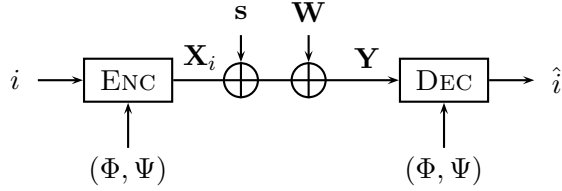

\subsection{Channel model}

The Gaussian AVC is shown in Figure \ref{fig:randgavc}.  For an input sequence $\mbf{x} \in \mathbb{R}^n$ the output of the Gaussian AVC is given by
	\begin{align*}
	\mbf{Y} = \mbf{x} + \mbf{s} + \mbf{W}.
	\end{align*}
The input is corrupted by iid additive white Gaussian noise
$\mbf{W}$ with variance $\sw$ and an unknown interference vector
$\mbf{s}$.  The input signal $\mbf{x}$ and jammer signal $\mbf{s}$ are
constrained in power:
	\begin{align*}
	\frac{1}{n} \norm{\mbf{x}}^2 &\le \xcostb \\
	\frac{1}{n} \norm{\mbf{s}}^2 &\le \scostb.
	\end{align*}
Define
	\begin{align}
	\mc{S}^n(\scostb) = \left\{ \mbf{s} : \norm{\mbf{s}}^2 \le n \scostb \right\}.
	\label{eq:allowS}
	\end{align}

An $(n,N)$ deterministic code $\mc{C}$ satisfying the input constraint $\xcostb$ is a pair of maps $(\denc,\ddec)$ with
	\begin{align*}
	\denc : [N] &\to \mathbb{R}^n \\
	\ddec : \mathbb{R}^n &\to [N],
	\end{align*}
such that for all $i \in [N]$ we have $\norm{\denc(i)}^2 \le n \xcostb$. 
An $(n,N)$ randomized code $\mbf{C}$ satisfying the input constraint $\xcostb$ is a random variable taking on values in the set of $(n,N)$ deterministic codes.  It is written as a pair of random maps $(\renc,\rdec)$ where each realization is an $(n,N)$ deterministic code satisfying the constraint $\xcostb$.  If $(\renc, \rdec)$ almost surely takes values in a set of $K$ codes, then we call this an $(n,N,K)$ randomized code.  The number $K$ is called the \textit{key size} of the randomized code.

In this section we will consider randomized coding and maximal probability of error:
	\begin{align}
	\stdmaxerr(\mbf{C},\mbf{s}) 
	&= \max_{i \in [N]} 
		\expe_{\mbf{C}}\left[
			\prob_{\mbf{W}}\left(
				\rdec( \renc(i) + \mbf{s} + \mbf{w}) \ne i 
			\right)
		\right] 
			\label{maxerr2} \\
	\stdmaxerr(\mbf{C}) 
	&= \max_{\mbf{s} \in \mc{S}^n(\scostb)} 
			\stdmaxerr(\mbf{C},\mbf{s}).
	\end{align}
A rate $R$ is achievable under maximal error with randomized coding if there exists a sequence of $(n,\lceil \exp(nR) \rceil)$ randomized codes whose maximal error goes to $0$ as $n \to \infty$.  The randomized coding capacity under maximal error $C_r(\xcostb,\scostb,\sw)$ is the supremum of the achievable rates under maximal error with randomized coding.  We will write $C_r$ when the parameters are clear.

\subsection{Main results}

Hughes and Narayan \cite{HughesN:87gavc} showed that if the input and
jammer are both bounded in power almost surely and the random variable
$\mbf{C}$ is unconstrained, then the capacity is equal to that of an
additive white Gaussian noise (AWGN) channel with the jammer treated
as additive noise:
	\begin{align}
	C_r(\xcostb,\scostb,\sw) = 
		\frac{1}{2} \log \left( 1 + \frac{\xcostb}{\scostb + \sw} \right).
	\label{eq:gaussian_cr}
	\end{align}
Csisz\'{a}r and Narayan \cite{CsiszarN:91gavc} showed that if only deterministic codes are allowed and the error criterion is replaced by the average probability of error:
	\begin{align*}
	\bar{\stdmaxerr}(\mbf{C}) = \frac{1}{N} \sum_{i=1}^{N} \prob_{\mbf{W}}\left(
				\ddec( \denc(i) + \mbf{s} + \mbf{w}) \ne i 
			\right), 
	\end{align*}
the capacity is equal to (\ref{eq:gaussian_cr}) if and only if the encoder has a higher
power limit than the jammer:
	\begin{align}
	\bar{C}_d(\xcostb,\scostb,\sw) = \left\{ 
		\begin{array}{ll}
		0 & 		\xcostb \le \scostb \\
		C_r(\xcostb,\scostb,\sw)   & \xcostb > \scostb.
		\end{array}
		\right.
	\label{eq:gaussian_cd}
	\end{align}
Recall that $f(n) = O(g(n))$ means there is a constant $c$ such that $f(n) \le c g(n)$ for sufficiently large $n$.  Our main result is to show that if $\log K(n) = O(\log n)$, for any $\capeps > 0$ the rate $C_r-\capeps$ is achievable using randomized codes with key size $K(n)$.  That is, $O(\log n)$ bits of common randomness is sufficient to achieve the randomized coding capacity of the GAVC.

\begin{theorem}
\label{thm:commrand}
The randomized coding capacity $C_r(\xcostb,\scostb,\sw)$ of the GAVC is achievable using randomized codes whose key size satisfies $\log K(n) = O(\log n)$.
\end{theorem}

\subsection{Analysis}

The class of randomized codes we consider can be built in two steps.
Similar to the discrete AVC construction in \cite{HughesT:96exponent}, we ``modulate'' a
single Gaussian codebook.  Let $N = \exp(n R)$
and $M$ be an arbitrary integer.
\begin{enumerate}
\item Let $\mc{B} = \{\mbf{x}_1, \mbf{x}_2, \ldots, \mbf{x}_N\}$ be
  a set of $N$ vectors on the sphere of radius $\sqrt{n \xcostb}$.  We can choose this set to have small maximal error for both the AWGN channel with noise variance $\scostb + \sw$ and the channel with additive noise $\mbf{V} + \mbf{W}$, where $\mbf{V}$ is uniform on the sphere of radius $\sqrt{n \scostb}$ and $\mbf{W}$ is iid Gaussian noise with variance $\sw$.
\item Let $\{U_k : k = 1, 2, \ldots, K\}$ be $n \times n$ unitary
  matrices generated uniformly from the set of all unitary matrices.
  Without loss of generality take $U_1 = I$.
\item The randomized code is uniform on the set $\{U_k \mc{B} : k = 1, 2, \ldots, K\}$.  To send message $i$, the encoder draws an integer $k$ uniformly from $\{1, 2, \ldots, K\}$ and encodes its message as $U_k \mbf{x}_i$.
\item The decoder knows $k$ and chooses the codeword in $\mc{B}_k$ that minimizes the distance to the received vector $\mbf{y}$:
	\begin{align*}
	\phi(\mbf{y},k) = \argmin_{j} \norm{ \mbf{y} - U_k \mbf{x}_j }.
	\end{align*}
\end{enumerate}

For all rates below $(1/2) \log(1 + \xcostb/(\scostb + \sw))$ we can choose the codebook $\mc{B}$ to have exponentially decaying probability of error \cite{Shannon:59gaussian,GallagerIT,Lapidoth:96additive} for both the AWGN channel with noise variance $\scostb + \sw$ and the channel with additive noise $\mbf{V} + \mbf{W}$:	
	\begin{align*}
	\stdmaxerr(\mc{B}) \le \exp( -n E( n^{-1} \log N ) ).
	\end{align*}
We can use this result to get a lower bound on the pairwise distance between any two codewords. 

Consider two codewords $\mbf{x}_i$ and $\mbf{x}_j$ from the set $\mc{B}$.  Let $\gamma > 0$ be half the distance between them:
	\begin{align*}
	\norm{ \mbf{x}_i - \mbf{x}_j } = 2 \gamma.
	\end{align*}
Suppose that we transmit $\mbf{x}_i$ over an AWGN channel with noise variance $\scostb + \sw$.  Then the probability of error for message $i$ can be lower bounded by the chance that the noise in the direction of $\mbf{x}_j - \mbf{x}_i$ is larger than $\gamma$.  Since the noise is iid, the error can be bounded by the integral of a Gaussian density \cite{TseV:wirelessbook}:
	\begin{align*}
	\stdmaxerr(i) 
	& \ge \frac{1}{\sqrt{2 \pi (\scostb + \sw)}} \int_{\gamma}^{\infty} \exp\left( - \frac{1}{2 (\scostb + \sw)} z^2 \right) dz \\
	&>
	\sqrt{ \frac{\scostb + \sw}{2 \pi \gamma^2} }
		\left( 1 - \frac{\scostb + \sw}{\gamma^2} \right) 
		\exp(- \gamma^2 / 2).
	\end{align*}
Therefore there exists a $\mu > 0$ such that for sufficiently large $n$ we have $\gamma > (\mu/2) \sqrt{n}$ for some $\mu > 0$, which means that 
	\begin{align}
	\norm{ \mbf{x}_i - \mbf{x}_j } > \mu \sqrt{n}.
	\label{eq:pairwisedist}
	\end{align}

We prove a more refined version of Theorem \ref{thm:commrand}.

\begin{theorem}
  \label{thm:gavc_noisy}
  Let $K(n)$ be chosen such that $K(n)/n \to \infty$ and $n^{-1} \log(K(n)/n) \to 0$.  For input power constraint $\xcostb$, jammer power constraint $\scostb$, and $\zeta > 0$ there is an $n$ sufficiently large and an $(n,N,K(n))$ randomized code for the GAVC of rate $R < C_r(\xcostb,\scostb,\sw)$, where
  \begin{align*}
	C_r(\xcostb, \scostb,\sw) = \frac{1}{2} \log\left(1 + \frac{\xcostb}{\scostb + \sw}\right),
  \end{align*}
whose error satisfies
	\begin{align*}
	\stdmaxerr(n) &= \zeta \frac{n}{K(n)}.
	\end{align*}
\end{theorem}

\begin{proof}
  Fix a rate $R < C_r$.  We will suppress the
  dependence of $K(n)$ on $n$ in the proof.  We need to show that for $n$
  sufficiently large, there exists a codebook $\mc{B}$ and $K$ unitary
  matrices $\{U_k\}$ such that the probability of error is bounded for
  any choice of $\mbf{s}$.  To do this we  first show that if
  $\mbf{s}$ lies in a dense subset of the $\sqrt{n\scostb}$ sphere,
  then the event that the average error for $K$ randomly chosen
  matrices $\{U_k\}$ is too large has probability exponentially small
  in $K$.  Therefore we can choose a collection $\{U_k\}$ that satisfies
  the probability of error bound for any $\mbf{s}$.
  
  Consider the codebook $\mc{B}$ of $N$ vectors from the sphere of 
  radius $\sqrt{n\xcostb}$.  The expected performance of
  this code is good for an additive noise channel with noise $\mbf{V} + \mbf{W}$,
  where $\mbf{V}$ is 
  distributed uniformly on the sphere of radius $\sqrt{n \scostb}$.
  That is, for any $\delta > 0$ there exists an $n$ sufficiently large
  such that
	\begin{align}
	\max_{i \in [N]} \expe_{\mbf{V},\mbf{W}}\left[ \stdmaxerr(i,\mbf{V}) \right] 
	< \exp( -n E(R)) \stackrel{\Delta}{=} \delta.
	\label{deltadef}
	\end{align}

Suppose that we sample $K$ points $\mbf{V}_1, \mbf{V}_2, \ldots,
\mbf{V}_K$ independently from the distribution of $\mbf{V}$.  Then
standard concentration bounds show that 
\begin{align}
  \prob_{\mbf{V},\mbf{W}} \left( \frac{1}{K} \sum_{k=1}^{K} \stdmaxerr(i,\mbf{V}_k) \ge t \right) 
	& \nonumber \\
	&\hspace{-1in}
	\le \exp\left(-K ( t \log \delta^{-1} - h_b(t) \log 2)\right).
  \label{expbnd}
\end{align}
A union bound over all $i \in [N]$ shows 
\begin{align}
  \prob_{\mbf{V}_n,\mbf{w}} \left( \bigcup_{i \in [N]}
  	\left\{
	\frac{1}{K} \sum_{k=1}^{K} \stdmaxerr(i,\mbf{V}_k) \ge t 
	\right \}
	\right) & \nonumber \\
  &\hspace{-2in}
  \le \exp\left(-K ( t \log \delta^{-1} - h_b(t) \log 2) + \log N\right).
  \label{eq:expbnd2}
\end{align}
Thus the probability that the collection of points $\{\mbf{v}_m\}$ induces an error probability that exceeds $t$ is exponentially small in $K$.

Now consider drawing $K$ unitary matrices $\{U_k : k = 1, 2, \ldots, K\}$ uniformly.  For a fixed $\mbf{v}$, the points
	\begin{align*}
	\mbf{V}_k = U_k^{-1} \mbf{v}
	\end{align*}
are uniform samples from $\mbf{V}$, and $\mbf{W}_k = U_k^{1} \mbf{W}$ are uniform samples from $\mc{N}(0,\sw I)$.  Let $\{ \mbf{a}_m : m = 1, 2, \ldots M\}$ be a set of vectors on the sphere of radius $\sqrt{n \scostb}$.  Another union bound yields the following:
	\begin{align}
	\prob\left( \bigcup_{m=1}^{M} \bigcup_{i \in [N]}
  	\left\{
	\frac{1}{K} \sum_{k=1}^{K} \stdmaxerr(i, U_{k}^{-1} \mbf{a}_m) \ge t 
	\right \}
	\right) 
	& \nonumber \\
	& \hspace{-2.3in}
	\le 
	\exp\left(-K ( t \log \delta^{-1} - h_b(t) \log 2) + \log M + \log N \right).
	\label{unionbnd}
	\end{align}

Results of Wyner \cite{Wyner:67packing} and Lapidoth \cite{Lapidoth:97mismatch} show that there exists a collection of $\exp(n (\rho + \epsilon))$ points on the $\sqrt{n \scostb}$-sphere
such that any point on the $\sqrt{n\scostb}$-sphere is at most a
distance $\eta$ from one of the points, where $\eta$ and $\rho$ are
related by $\rho = (1/2) \log\left(\scostb/\eta^2 \right)$.
Choose $M = \exp(n (\rho + \epsilon))$ and let
$\{\mbf{a}_m\}$ be the corresponding rate-distortion codebook.  The bound (\ref{unionbnd}) implies
	\begin{align}
	\prob\left( \bigcup_{m=1}^{M} \bigcup_{i \in [N]}
  	\left\{
	\frac{1}{K} \sum_{k=1}^{K} \stdmaxerr(i, U_{k}^{-1} \mbf{a}_m) \ge t 
	\right \}
	\right) 
	& \nonumber \\
	& \hspace{-2.3in}
	\le
	\exp\left(-K ( t \log \delta^{-1} - h_b(t) \log 2) + n (\rho + R + \epsilon ) \right).
	\label{densebnd}
	\end{align}
If $K(n)/n \to \infty$ then the probability that the error is
smaller than $t$ for the $M$ points $\{\mbf{a}_m\}$ can be made arbitrarily close to $1$ for any $\eta$.  The next step is to argue that we can extend the bound from $\mbf{s} \in \{\mbf{a}_m\}$ to all $\mbf{s}$.

Because $R < C_r(\xcostb,\scostb,\sw)$, for a sufficiently small constant $\nu$ we have 
	\begin{align*}
	R < \frac{1}{2} \log\left( 1 + \frac{\xcostb}{(1 + \nu)^2 \scostb + \sw} \right), 
	\end{align*}
for some sufficiently small constant $\nu$.  That is, we can choose our code to have small error probability for noise of variance $(1 + \nu)^2 \scostb + \sw$.  The bound in (\ref{densebnd}) shows that for each message $i$ there is a set $\mc{K}_i$ of at least $(1-t) K$ keys for which  
	\begin{align*}
	\norm{ \mbf{x}_i - \mbf{x}_j + (1 + \nu) U_k^{-1} \mbf{a}_m } 
	> (1 + \nu) \| \mbf{a}_m \| \qquad \forall j \ne i.
	\end{align*}
Equivalently, we can write
	\begin{align}
	2 \ip{ \mbf{x}_j - \mbf{x}_i }{ U_k^{-1} \mbf{a}_m }
	< \frac{1}{1 + \nu} \norm{ \mbf{x}_i - \mbf{x}_j }^2.
	\label{eq:ipmindistbound}
	\end{align}
Now suppose the jammer's input is $\mbf{s}$.  For each $k \in \mc{K}_i$, the rate-distortion codebook property guarantees that $\mbf{s}$ is only a distance $\eta \sqrt{n}$ from some point $U_k^{-1} \mbf{a}_m$.  We would like to prove a bound like (\ref{eq:ipmindistbound}) for all $\mbf{s}$.  To start:
	\begin{align*}
	2 \ip{ \mbf{x}_j - \mbf{x}_i }{ \mbf{s} }
	&= 2 \ip{ \mbf{x}_j - \mbf{x}_i }{ \mbf{s} - U_k^{-1} \mbf{a}_m } 
		\\
		& \hspace{1in} + 2 \ip{ \mbf{x}_j - \mbf{x}_i }{ U_k^{-1} \mbf{a}_m } \\
	&< 2 \cdot \norm{ \mbf{x}_i - \mbf{x}_j } \cdot \eta \sqrt{n}
		+ \frac{1}{1 + \nu} \norm{ \mbf{x}_i - \mbf{x}_j }^2 \\
	&< \left( 2 \frac{\eta}{\mu} + \frac{1}{1 + \nu} \right)
		\norm{ \mbf{x}_i - \mbf{x}_j }^2,
	\end{align*}
where we used (\ref{eq:ipmindistbound}), the Cauchy-Schwartz inequality, the distortion bound for $\{\mbf{a}_m\}$, and and (\ref{eq:pairwisedist}).  Now choose $\eta$ sufficiently small so that 
	\begin{align}
	2 \ip{ \mbf{x}_j - \mbf{x}_i }{ \mbf{s} }
	< \norm{ \mbf{x}_i - \mbf{x}_j }^2.
	\end{align}
This shows that the minimum distance decoding rule results in a small error probability for all $\mbf{s}$ with $\norm{\mbf{s}}^2 = n \scostb$.
	
The last thing we need is to show that the average error probability
is monotonic in the length of the jamming vector for a given
direction.  Suppose that there was an error for $\mbf{s}$ but now the
jammer inputs $(1 + b) \mbf{s}$.  Then there is an error if
	\begin{align*}
	\norm{ \mbf{x}_i - \mbf{x}_k + (1 + b) \mbf{s} } \le (1 + b) \norm{\mbf{s}}.
	\end{align*}
\noindent But we can easily bound this using the triangle inequality:
	\begin{align*}
	\norm{ \mbf{x}_i - \mbf{x}_k + (1 + b) \mbf{s} } &\le \norm{ \mbf{x}_i - \mbf{x}_k + \mbf{s}} + \norm{b \mbf{s}} \\
	&\le (1 + b) \norm{\mbf{s}}.
	\end{align*}
\noindent Thus the error probability can only become smaller for 
shorter jamming inputs $\mbf{s}$.  

We have shown that for any $t > 0$ there is an $n$ sufficiently large such that with high probability, choosing a random set of $K$ unitary matrices $\{\mbf{U}_k\}$ results in a randomized code whose error can be made smaller than $t$.  
Therefore such a randomized code exists.
\end{proof}

\subsection{An application to degraded broadcast}

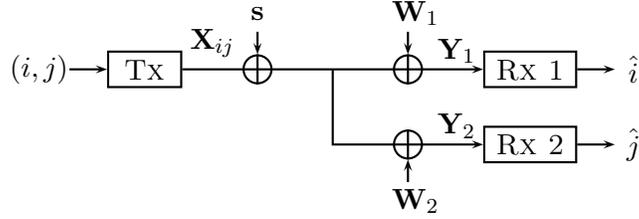
\begin{figure}
\begin{center}
\psset{xunit=0.05cm,yunit=0.05cm,runit=0.05cm}
\begin{pspicture}(-5,-10)(170,40)
\rput(2,30){$(i,j)$}
\psline{->}(10,30)(20,30)
\psframe(20,25)(40,35)
\rput(30,30){\textsc{Tx}}
\rput(48,37){$\mbf{X}_{ij}$}

\pscircle(60,30){4}
\psline(60,26)(60,34)
\psline{->}(60,40)(60,34)
\rput(60,45){$\mbf{s}$}

\psline{->}(40,30)(120,30)
\psline{->}(80,30)(80,10)(120,10)

\pscircle(100,30){4}
\psline(100,26)(100,34)
\psline{->}(100,40)(100,34)
\rput(102,45){$\mbf{W}_1$}

\pscircle(100,10){4}
\psline(100,6)(100,14)
\psline{->}(100,0)(100,6)
\rput(102,-5){$\mbf{W}_2$}

\rput(113,35){$\mbf{Y}_1$}
\psframe(120,25)(145,35)
\rput(132.5,30){\textsc{Rx 1}}
\psline{->}(145,30)(155,30)
\rput(160,30){$\hat{i}$}

\rput(113,15){$\mbf{Y}_2$}
\psframe(120,5)(145,15)
\rput(132.5,10){\textsc{Rx 2}}
\psline{->}(145,10)(155,10)
\rput(160,10){$\hat{j}$}

\end{pspicture}
\caption{The arbitrarily varying degraded Gaussian broadcast channel.  The jammer is shared between both receivers.   Since we assume the noise $\mbf{w}_2$ has higher variance that $\mbf{w}_1$, we call receiver $1$ the ``strong'' user and receiver 2 the ``weak'' user.  \label{fig:avbcfig}}
\end{center}
\end{figure}

We can apply Theorem \ref{thm:gavc_noisy} to the a degraded broadcast
channel with a common jammer.  In this setting we consider
deterministic coding.  We show that one receiver can use the codeword of
the other receiver to enable randomized coding for its message.  The
channel model is given by
	\begin{align*}
	\mbf{Y}_1 &= \mbf{x} + \mbf{s} + \mbf{W}_1 \\
	\mbf{Y}_2 &= \mbf{x} + \mbf{s} + \mbf{W}_2,
	\end{align*}
where $\mbf{W}_1$ is iid Gaussian with variance $\sigma_1^2$, $\mbf{W}_2$ is iid Gaussian with variance $\sigma_2^2$, and $\sigma_1^2 < \sigma_2^2$.  The channel is shown in Figure \ref{fig:avbcfig}.  We call receiver $1$ the strong user and receiver $2$ the weak user.

An $(n, N_1, N_2)$ deterministic code with power constraint $\xcostb$ for this channel is a tuple of maps $(\denc,\ddec_1,\ddec_2)$, where 
	\begin{align*}
	\denc &: [N_1] \times [N_2] \to \mathbb{R}^n \\
	\ddec_1 &: \mathbb{R}^n \to [N_1] \\
	\ddec_2 &: \mathbb{R}^n \to [N_2],
	\end{align*}
and $\norm{ \denc(i,j) }^2 \le n \xcostb$ for all $(i,j)$.  The map $\denc$ is the encoder and the maps $\ddec_1$ and $\ddec_2$ are the decoders for users $1$ and $2$.  The average probability of error for the code under state constraint $\scostb$ is
	\begin{align*}
	\avgerr &= \max_{\mbf{s} \in \mc{S}^n(\scostb)}
		\frac{1}{N_1 N_2} \sum_{i=1}^{N_1} \sum_{j=1}^{N_2} 
		\prob \big( 
			\ddec_1( 
				\denc(i,j) + \mbf{s} + \mbf{W}_1) \ne i, \\
		&\hspace{1.8in}
				\ddec_2( \denc(i,j) + \mbf{s} + \mbf{W}_2) \ne j
			\big).
	\end{align*}
The error is averaged over the messages to both users.  We say the pair of rates $(R_1, R_2)$ is achievable if there exists a sequence of $(n, \exp(nR_1), \exp(n R_2))$ deterministic codes whose average error goes to $0$ as $n \to \infty$.  The capacity region is the union of achievable rates.

The discrete arbitrarily varying broadcast channel without constraints was first studied by Jahn \cite{Jahn:81multiuser}, who proved an achievable rate region for randomized coding and then applied the elimination technique \cite{Ahlswede:78elimination} to derandomize the code.  This approach does not work in general for constrained AVCs.  Discrete constrained AVCs with degraded message sets were studied by Hof and Bross \cite{HofB:06avbcast}.  Their achievable strategy requires a number of non-symmetrizability conditions which are analogous to our result in Theorem \ref{thm:gavdbc}.  %

We build a superposition code \cite{Cover:72bcast} based on our rotated codebook construction.  The strong user can treat the message for the weak user as a random key in a randomized code.  The codebook for user 2 is a deterministic code (``cloud centers'') with power $\alpha \Gamma$ and the codebook for user 1 is a randomized code with power $(1 - \alpha) \Gamma$, where the randomization is over the codewords of user 2.  From Theorem \ref{thm:gavc_noisy} we can see that the randomization provided by user 2's message is sufficient for user 1 to achieve the randomized coding capacity. This scheme is limited \cite{CsiszarN:91gavc} to those $\alpha$ for which $\alpha \xcostb > \scostb$.

\begin{theorem}		\label{thm:gavdbc}
  If $\scostb \ge \xcostb$ then the deterministic coding capacity region of the arbitrarily varying
  degraded Gaussian broadcast channel is the empty set.  If $\scostb \le \xcostb$ then for $\alpha \in (\scostb/\xcostb, 1]$, the rates $(R_1, R_2)$ satisfying the following inequalities are
  achievable with deterministic codes for the arbitrarily varying
  degraded Gaussian broadcast channel under average probability of
  error:
	\begin{align}
	R_1 &< \frac{1}{2} \log\left(
		1 + \frac{(1 - \alpha) \xcostb
			}{
			\scostb + \sigma_1^2}
			\right) \label{ratestrong} \\
	R_2 &< \frac{1}{2} \log\left(
		1 + \frac{\alpha \xcostb
			}{
			(1-\alpha) \xcostb + \scostb + \sigma_2^2}
			\right) \label{rateweak} \\
	R_1 + R_2 
		&< \frac{1}{2} \log\left(
			1 + \frac{\xcostb - \scostb
			}{
			\scostb + \sigma_1^2}
			\right)
			+
			\frac{1}{2} \log\left(
		1 + \frac{\scostb
			}{
			\xcostb + \sigma_2^2}
			\right).
			\label{eq:ratesplit}
	\end{align}
\end{theorem}

\begin{proof}
The converse follows from the converse for the standard AVC.  Since we are limited to deterministic codes, if $\scostb \ge \xcostb$ the jammer can choose a message pair $(i',j')$ and transmit $\denc(i',j')$ plus additional noise.

To show the achievable rate region, suppose that $\scostb < \xcostb$ and generate a codebook using Theorem \ref{thm:gavc_noisy} containing $N_2 = \exp(n R_2)$ codewords $\{\mbf{v}_j\}$ on the $\sqrt{n \alpha \xcostb}$-sphere.  For each $\mbf{u}_i$ generate $N_1 = \exp(n R_1)$ codewords $\{\mbf{v}_{ij}\}$ uniformly on the $\sqrt{(n-1)(1 - \alpha) \xcostb}$-sphere.  Let the overall codebook be:
	\begin{align*}
	\mbf{x}_{ij} = \mbf{u}_i + A_{i} U_i \mbf{v}_{ij},
	\end{align*}
where $A_i$ is an isometric mapping of $\mathbb{R}^{n-1}$ to the plane orthogonal to $\mbf{u}_i$ and $U_i$ is a random unitary transformation as in the construction of Theorem \ref{thm:gavc_noisy}.  The codebooks satisfy the the power constraint.

The weak decoder first decodes $\mbf{u}_i$, treating the signal $A_{i} U_i \mbf{v}_{ij}$ as additional noise.  From the results of Csisz\'{a}r and Narayan \cite{CsiszarN:91gavc} we know that the average probability of error can be made small if $\alpha \xcostb > \scostb $.  This gives the first rate bound.

The strong decoder replicates the first step of the weak user.  If message $i$ was decoded correctly, it can subtract out $\mbf{u}_i$ and the residual channel is identical to a GAVC with input power $(1 - \alpha) \xcostb$ using the codebook of Theorem \ref{thm:gavc_noisy}.  This gives us the second rate bound.

To see the sum-rate bound (\ref{eq:ratesplit}), note that the weak user can give up any part of its message to the strong user, which means that rate splitting between the points where $\alpha = \scostb/\xcostb$ and $\alpha = 0$ are also achievable.
\end{proof}

\begin{figure} 
\centering
\subfigure[Complete region]{
	\includegraphics[width=2.5in]{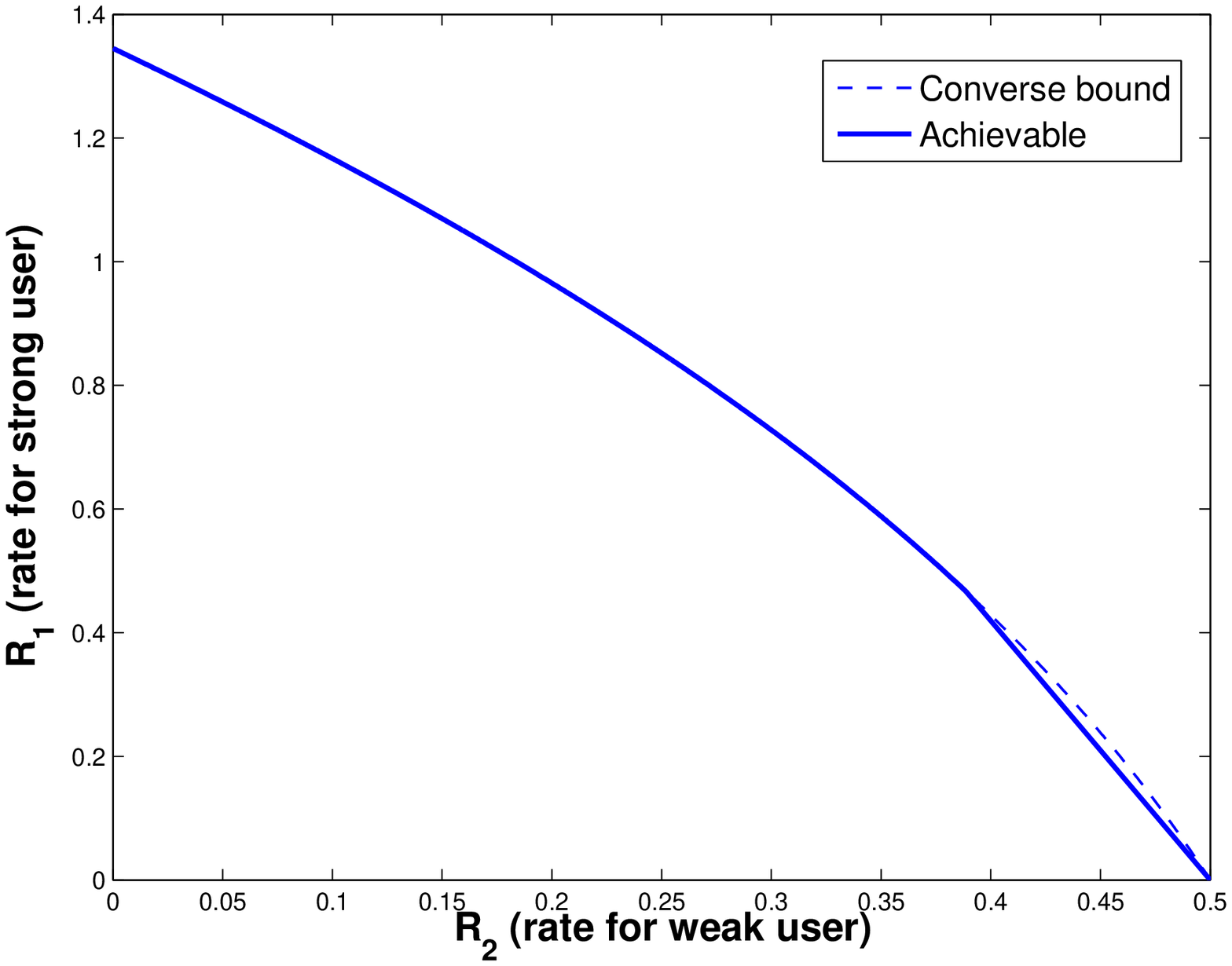} 
}
\subfigure[Detail of gap]{
	\includegraphics[width=2.5in]{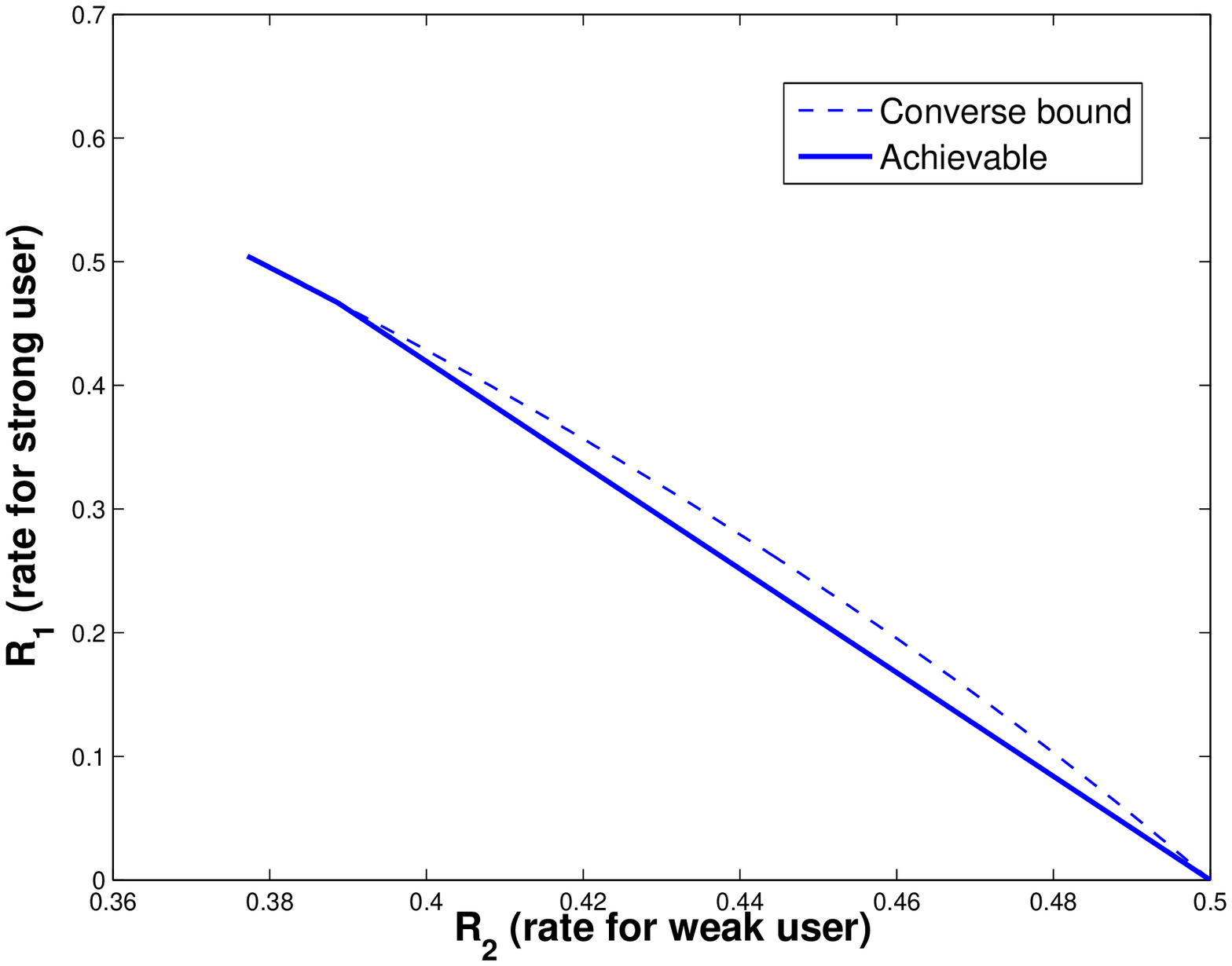} 
}
\caption{Achievable rates for the degraded broadcast Gaussian AVC with $\xcostb = 6$, $\scostb = 1$, $\sigma_1^2 = 0.1$, and $\sigma_2^2 = 5$. } 
\label{fig:dbc_bounds} 
\end{figure}

A plot of the achievable rate region is shown in Figure \ref{fig:dbc_bounds}.  This achievable region is tight for $\alpha > \scostb/\xcostb$ because the jammer could just add Gaussian noise to make the channel a degraded Gaussian broadcast channel \cite{Bergmans:73bcast,Bergmans:74converse,BergmansC:74bcast,Gallager:74bcast}.  The coding scheme above cannot be used in the regime where $\alpha \le \scostb/\xcostb$ because the jammer can symmetrize the $\{\mbf{u}_j\}$ codebook to the stronger user.  At present we do not know if new achievable strategies can achieve higher rates in this regime or if different converse arguments can show that rate splitting is optimal.

\section{Dirty paper coding for AVCs}

In this section we turn to a different AVC model in which there are two sources of interference, one of which is known to the transmitter.  The benefits of channel state information at the transmitter have been investigated by researchers since Shannon \cite{Shannon:58sideinfo}.  In one version of the problem, a time-varying state sequence is known non-causally at the transmitter, and the encoder can base its codebook on this known sequence.  The capacity for discrete channels with iid state sequences was found in the celebrated paper of Gel'fand and Pinsker \cite{GelfandP:80parameters}.  Costa \cite{Costa:83dirty} showed an analogous result for the Gaussian case and showed that the capacity is equal to that of a channel with no interference at all.  His strategy is called a ``dirty paper code.''  These results have found applications to inter-symbol interference (ISI) channels \cite{ErezSZ:05known}, watermarking \cite{CohenL:02watermark}, multi-antenna broadcasting \cite{WeingartenSS:06mimobc}, and models for ``cognitive radio'' \cite{DevroyeMT:06cognitive,JovicicV:06cognitive}.  We show that additional interference can increase the capacity of the GAVC when the encoder and decoder do not share any common randomness.

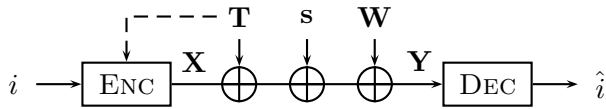
\begin{figure}
\begin{center}
\psset{xunit=0.06cm,yunit=0.06cm,runit=0.06cm}
\begin{pspicture}(0,20)(140,50)
\rput(5,30){$i$}
\psline{->}(10,30)(20,30)
\psframe(20,25)(40,35)
\rput(30,30){\textsc{Enc}}
\psline{->}(40,30)(100,30)
\rput(45,35){$\mathbf{X}$}

\rput(55,30){
	\begin{pspicture}(-4,-4)(4,4)
	\pscircle(0,0){4}
	\psline(-4,0)(4,0)
	\psline(0,-4)(0,4)
	\end{pspicture}
	}
\psline{->}(55,40)(55,34)
\rput(55,45){$\mathbf{T}$}

\psline[linestyle=dashed,arrows=->](50,45)(30,45)(30,35)

\rput(70,30){
	\begin{pspicture}(-4,-4)(4,4)
	\pscircle(0,0){4}
	\psline(-4,0)(4,0)
	\psline(0,-4)(0,4)
	\end{pspicture}
	}
\psline{->}(70,40)(70,34)
\rput(70,45){$\mathbf{s}$}

\rput(85,30){
	\begin{pspicture}(-4,-4)(4,4)
	\pscircle(0,0){4}
	\psline(-4,0)(4,0)
	\psline(0,-4)(0,4)
	\end{pspicture}
	}
\psline{->}(85,40)(85,34)
\rput(85,45){$\mathbf{W}$}

\rput(95,35){$\mathbf{Y}$}
\psframe(100,25)(120,35)
\rput(110,30){\textsc{Dec}}

\psline{->}(120,30)(130,30)
\rput(135,30){$\hat{i}$}
\end{pspicture}
\caption{The Gaussian arbitrarily varying channel with a known interference signal at the encoder.  \label{randavc}}
\end{center}
\end{figure}

\subsection{Channel model}

We will consider channels with inputs and outputs in $\mathbb{R}^n$ of the form
	\begin{align}
	\mbf{Y} = \mbf{X} + \mbf{T} + \mbf{s} + \mbf{W}.
	\label{eq:chandef}
	\end{align}
Here we take $\mbf{W} \sim \mc{N}(0, \sw I)$, $\norm{\mbf{s}}^2 \le \scostb n$, $\norm{\mbf{X}}^2 \le \xcostb n$, and $\mbf{T} \sim \mc{N}(0, \st I)$.  The channel input created by the transmitter is $\mbf{X}$, the vector $\mbf{T}$ is interference known to the transmitter,  $\mbf{s}$ is jamming interference, and $\mbf{W}$ is the independent noise at the receiver.  If randomized coding is allowed, then Costa's result implies that the capacity is equal to the AWGN capacity without $\mbf{T}$ and the jammer treated as additional noise.

An $(n,N)$ code with power constraint $\Gamma$ for this channel is a pair of functions $(\denc, \ddec)$, where $\denc : [N] \times \mathbb{R}^n \to \mathbb{R}^n$ and $\ddec : \mathbb{R}^n \to [N]$ and 
	\begin{align*}
	\norm{\denc( i, \mbf{T} )}^2 \le n \xcostb  \qquad a.s.. 
	\end{align*}
The average probability of error (over $\mbf{W}$ and $\mbf{T}$) for this code with jammer power $\scostb$ is given by
	\begin{align*}
	\bar{\varepsilon} = \max_{\mbf{s} \in \mc{S}^n(\scostb)}
	\frac{1}{N} \sum_{i=1}^{N} \prob \left( \ddec( \denc(i,\mbf{T}) + \mbf{T} + \mbf{s} + \mbf{W} ) \ne i \right).
	\end{align*}
	
A rate $R$ is \textit{achievable} if there exists a sequence of $(n, \lceil \exp(nR) \rceil)$ codes with $\bar{\varepsilon}_n \to 0$ as $n \to \infty$.  The \textit{capacity} $\bar{C}_d$ is defined to be the supremum of all achievable rates.  For $\st = 0$ this channel model reduces to the Gaussian AVC \cite{CsiszarN:91gavc} whose capacity is given in \eqref{eq:gaussian_cd}.

\subsection{Main result}

Our main result is an achievable rate region for the Gaussian AVC with partial state information at the encoder that is achievable using a generalized dirty-paper code.  For some parameter values the achievable rate is the capacity of the channel.  One way of interpreting this result is that the presence of extra interference known to the transmitter boosts its effective power and therefore lowers the power threshold for the standard Gaussian AVC.

\begin{theorem}
	\label{thm:dpc_thm}
Let
	\begin{align*}
	\mc{A}(\scostb) &= \left\{ (\boost,\corr) : 
		\frac{ 
		\left( \xcostb + (1 + \boost) \corr \sqrt{\xcostb \st} + \boost \st \right)^2
		}{ 		
			\xcostb + 2 \corr \boost \sqrt{\xcostb \st} + \boost^2 \st
		} 
		> \scostb
		\right\} \\
	P_U &= \xcostb + 2 \corr \boost \sqrt{\xcostb \st} + \boost^2 \st \\
	P_I &= \scostb + \sw \\
	P_Y &= \xcostb + 2\corr \sqrt{\xcostb \st} + \st + \scostb + \sw.
	\end{align*}
The following rate is achievable:
	\begin{align}
	R &= \max_{(\boost,\corr) \in \mc{A}(\scostb)} \frac{1}{2} \log \left( 
		\frac{ ( 1 - \corr^2 ) \xcostb P_Y
		}{ (1 - \boost)^2 ( 1 - \corr^2 )  \xcostb \st
			+ P_I P_U }
		\right).
	\label{eq:achR}
	\end{align}
\end{theorem}

In Costa's original paper, choosing $\corr = 0$ and $\boost = \optboost$, where
	\begin{align}
	\optboost = \frac{\xcostb}{ \xcostb + \scostb + \sw }
	\label{eq:dpc:optboost}
	\end{align}
gives an achievable rate of $(1/2) \log( 1 + \xcostb/(\scostb + \sw) )$.  We have
a simple corollary to show when our general scheme achieves capacity.

\begin{corollary}[Capacity achieving parameters]
If $\xcostb$, $\scostb$ and $\st$ are such that
	\begin{align}
	\scostb  < \frac{ (\xcostb + \optboost \st)^2 }{ \xcostb + \optboost^2 \st },
	\label{eq:st_thresh}
	\end{align}
then the capacity of the channel \eqref{eq:chandef} under deterministic coding is
	\[
	\bar{C}_d = \frac{1}{2} \log \left(  1 + \frac{\xcostb}{\scostb + \sw} \right),
	\]
and is achievable using the dirty paper code.
\label{cor:cap}
\end{corollary}

There is a threshold on $\scostb$ making the capacity equal to $0$.  If the jammer can simulate both the known interference and the transmitter's strategy, then it can symmetrize the channel.  The proof is straightforward and omitted.

\begin{lemma}[Na\"{i}ve outer bound]
We have
	\begin{align*}
	\bar{C}_d \le \left\{ 
		\begin{array}{ll}
		0 & 		\scostb > (\stone + \sqrt{\xcostb})^2 \\
		\frac{1}{2} \log \left( 1 + \frac{\xcostb}{\scostb + \sw} \right)
			& \textrm{otherwise}
		\end{array}
		\right. .
	\end{align*}
\label{lem:outerbnd}
\end{lemma}

\subsection{Analysis}

Our codebook construction uses two auxiliary rates $R_U$ and $\Rbin$ and will depend on parameters $\boost$ and $\corr$ to be chosen later and positive constants $\rateeps$ and $\whpeps$ that can be made arbitrarily close to $0$. 
\begin{enumerate}
\item The encoder will generate an auxiliary codebook $\{ \mbf{U}_j \}$ of $\exp(n (R_U - \rateeps))$ vectors drawn uniformly from the $n$-sphere of power $P_U$.
\item These codewords are divided randomly into $\exp( n (R - 2 \rateeps))$ bins $\{\mc{B}_m\}$ such that each bin has $\exp( n (\Rbin + \rateeps))$ codewords.  We denote the $i$-th codeword of bin $\mc{B}_m$ by $\mbf{U}(m,i)$.
\item Given a message $m$ and an interference vector $\mbf{T}$, the encoder chooses the vector $\mbf{U}(m,i) \in \mc{B}_m$ that is closest to $\beta \mbf{T}$, where
	\begin{align*}
	\beta^2 = \frac{P_U}{\st} \left( 1 + \frac{ (1 - \corr^2) \xcostb }{ P_U - (1 - \corr^2) \xcostb } \right).
	\end{align*}
If no such $\mbf{U}(m,i)$ exists then we declare an encoder error.  The encoder transmits
	\begin{align*}
	\mbf{X} = \mbf{U}(m,i) - \boost \mbf{T}.
	\end{align*}
We will show that for $\whpeps > 0$, we can choose $n$ sufficiently large such that
	\begin{align}
	\norm{\mbf{X}}^2 
		& \le \xcostb
		\label{eq:inputbnd} \\
	\ip{ \mbf{U}(m,i) - \boost \mbf{T} }{ \mbf{T}} 
    		&\ge \corr \sqrt{\xcostb \st} - \whpeps.
		\label{eq:setcorr}
	\end{align}
\item The decoder first attempts to decode $\mbf{U}(m,i)$ out of the overall codebook $\{ \mbf{U}_j \}$ and produces an estimate $\mbf{U}(\hat{m},\hat{i})$.  It then outputs the estimated message index $\hat{m}$.
\end{enumerate}

We will analyze the performance of this coding strategy on an AVC with for general $\corr$ and $\boost$.

\begin{lemma}
\label{lem:good_enc}
Suppose
	\begin{align}
	\Rbin \ge \frac{1}{2} \log \left(
		\frac{ P_U }{ (1 - \corr^2) \xcostb }
		\right).
	\label{eq:dpc:succ_enc}
	\end{align}
Then for any $\rateeps > 0$ and  $\whpeps > 0$ in the code construction and any $\epsilon' > 0$, there exists an $n$ sufficiently large such that 
	\begin{align*}
	\prob\left( 
		\exists \mbf{U}(m,i) \in \mc{B}_m : \eqref{eq:inputbnd}, \eqref{eq:setcorr}\ \textrm{hold} \right)
    \ge
    1 - \epsilon'.
	\end{align*}
\end{lemma}

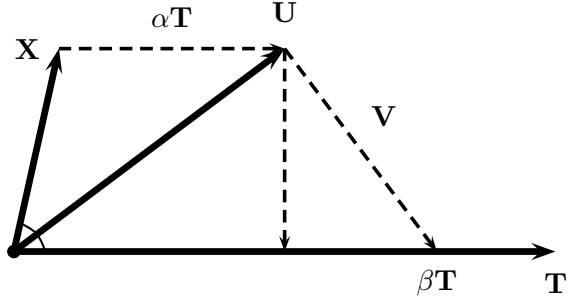
\begin{figure}
\begin{center}

\psset{unit=0.06cm}

\begin{pspicture}(-5,-10)(130,70)

\psline[linewidth=1.5,arrows=->](0,0)(60,45)
\rput(60,53){$\mbf{U}$}

\psline[linewidth=1.5,arrows=->](0,0)(120,0)
\rput(120,-7){$\mbf{T}$}

\psline[linewidth=0.8,linestyle=dashed,arrows=<-](60,45)(10,45)
\rput(35,52){$\alpha \mbf{T}$}

\psline[linewidth=1.5,arrows=->](0,0)(10,45)
\rput(3,45){$\mbf{X}$}

\psline[linewidth=0.8,linestyle=dashed,arrows=->](60,45)(93.75,0)
\rput(82,30){$\mbf{V}$}
\rput(93.75,-7){$\beta \mbf{T}$}

\psline[linewidth=0.8,linestyle=dashed,arrows=->](60,45)(60,0)

\pscurve(2,6)(5,4)(7,0)
\pscircle*(0,0){1.5}

\end{pspicture}
\caption{Geometric picture for dirty-paper encoding with general parameters.  \label{fig:dpc_cartoon}}
\end{center}
\end{figure}

\begin{proof}
Consider the picture in Figure \ref{fig:dpc_cartoon} and let
	\begin{align*}
	\beta^2 = \frac{P_U}{\st} \left( 1 + \frac{ (1 - \corr^2) \xcostb }{ P_U - (1 - \corr^2) \xcostb } \right).
	\end{align*}
We must show that a $\mbf{U}(m,i) \in \mc{B}_m$ exists satisfying \eqref{eq:setcorr}.  By the rate-distortion theorem for Gaussian sources, for $\Rbin$ satisfying \eqref{eq:dpc:succ_enc}, the codebook $\mc{B}_m$ chosen uniformly on the sphere of power $P_U$ can compress the source $\beta \mbf{T}$ to distortion 
	\begin{align*}
	D = \frac{ P_U (1 - \corr^2) \xcostb }{ P_U - (1 - \corr^2) \xcostb }.
	\end{align*}
To see this, consider the test channel $\beta \mbf{T} = \mbf{U} + \mbf{V}$, where $\mbf{V}$ is iid Gaussian with variance $D$.  The mutual information of this test channel is
	\begin{align*}
	\frac{1}{n} \mi{ \beta \mbf{T} }{ \mbf{U} }
	&=
		\frac{1}{2} \log \left( 1 + \frac{ P_U }{ D } \right) \\
	&= \frac{1}{2} \log \left(
		\frac{ P_U }{ (1 - \corr^2) \xcostb }
		\right).
	\end{align*}

We can choose $\mbf{U}(m,i)$ to be the codeword in $\mc{B}_m$ that corresponds to quantizing $\beta \mbf{T}$.  For any $\rateeps > 0$ the codebook $\mc{B}_m$ has rate greater than the rate distortion function for the source $\beta \mbf{T}$ with distortion $D$, so there exists an $\epsilon > 0$ such that with high probability,
	\begin{align*}
	\norm{ \beta \mbf{T} - \mbf{U}(m,i) }^2 \le D - \epsilon.
	\end{align*}
For any $\delta > 0$ we can choose $n$ sufficiently large such that with high probability we have
	\begin{align*}
	\ip{ \mbf{U} }{ \mbf{T} } 
		>
		\sqrt{ \frac{ P_U \st }{ P_U + D - \epsilon }	}
		- \delta.
	\end{align*}
If we choose $\delta$ small enough, we can find an $\eta > 0$ such that with high probability,
	\begin{align*}
	\ip{ \mbf{U} }{ \mbf{T} } 
		&>
		\sqrt{ \frac{ P_U \st }{ P_U + D } }
		+ \eta \\
	&= \sqrt{ \st (P_U - (1 - \corr^2) \xcostb) } + \eta \\
	&= \corr \sqrt{\xcostb \st} + \boost \st + \eta.
	\end{align*}
Therefore with high probability we have
	\begin{align*}
	\ip{ \mbf{U} - \boost \mbf{T} }{ \mbf{T} } > \corr \sqrt{\xcostb \st} + \eta/2.
	\end{align*}
Therefore we can choose $\whpeps$ to satisfy \eqref{eq:setcorr}.

The last thing to check is that we can satisfy the encoder power constraint with high probability.  Setting $\mbf{X} = \mbf{U} - \boost \mbf{T}$, we can see that for any $\delta' > 0$, with high probability
	\begin{align*}
	\norm{ \mbf{X} }^2 &\le P_U - 2 \boost \ip{ \mbf{U} }{ \mbf{T} } + \boost^2 \st + \delta'.
	\end{align*}
Choosing $\delta'$ sufficiently small yields $\norm{ \mbf{X} }^2 < \xcostb - \boost \eta$, which proves the result.
\end{proof}

\begin{proof}[Proof of Theorem \ref{thm:dpc_thm}]
We will choose the constants $\rateeps$ and $\whpeps$ according to Lemma \ref{lem:good_enc}.  The decoder must decode $\mbf{U}(m,i)$ from the received signal $\mbf{Y}$:
	\begin{align}
	\mbf{Y} &= \mbf{U}(i) + (1 - \boost) \mbf{T} + \mbf{s} + \mbf{W}.
	\label{eq:dpc:recsignal}
	\end{align}
The codebook $\{\mbf{U}_k\}$ can be used to achieve any rate below the deterministic coding capacity of the GAVC with input $\mbf{U}$, noise $\mbf{W} + (1 - \boost) \mbf{T}$, and jamming interference $\mbf{s}$, provided $P_U > \scostb$.  We can therefore choose $R_U$ to be equal to this capacity and for fixed $\boost$ and $\corr$ we calculate the capacity in what follows.  

We first find the power of the component of $\mbf{T}$ that is orthogonal to $\mbf{U}$:
	\begin{align}
	\mbf{T} = \frac{ \ip{\mbf{U}}{\mbf{T}} }{ \norm{\mbf{U}}^2 } \mbf{U} 
		+ \left( \mbf{T} - \frac{ \ip{\mbf{U}}{\mbf{T}} }{ \norm{\mbf{U}}^2 } \mbf{U} \right).
	\label{eq:Tdecompose}
	\end{align}
From \eqref{eq:setcorr} we see that for any $\delta > 0$ we can choose $n$ sufficiently large that
	\begin{align*}
	\prob \left(
		\ip{\mbf{U}}{\mbf{T}} 
		\ge \corr \sqrt{\xcostb \st} - \boost \st - 2 \whpeps
	\right)
	\ge 1 - \delta.
	\end{align*}
Let $P_T$ be the expected power in the second term of \eqref{eq:Tdecompose}.  Then for sufficiently large $n$ we also have
	\begin{align*}
	\prob \left(
		P_T \le \left( \st -  \frac{\left( \corr \sqrt{\xcostb \st} + \boost \st - 2 \whpeps \right)^2}{ P_U  }
		\right)
		\right)
	\ge 1 - \delta.
	\end{align*}
Some algebraic manipulation reveals that there is a constant $c$ such that 
	\begin{align*}
	\prob \left(
		P_T -
			\frac{(1 - \corr^2) \xcostb \st}{ P_U } 
		\le c \whpeps
		\right)
	\ge 1 - \delta.
	\end{align*}
In the GAVC \eqref{eq:dpc:recsignal} we define the equivalent noise variance as $P_I + (1 - \boost)^2 P_T$.  

In order for $\mbf{U}$ to be decodable, $R_U$ must be smaller than the
capacity of the corresponding AWGN channel:
	\begin{align*} 
	R_U < \frac{1}{2} \log \left( 
		\frac{ P_U P_Y }{ (1 - \boost)^2 (1 - \corr^2) \xcostb \st + P_I P_U  }
		\right).
	\end{align*}
Then $R_U - \Rbin$ gives the term to be maximized in \eqref{eq:achR}.
Note that in the presence of a jammer with power constraint $\scostb$, the $\mbf{U}$ codebook is only capacity achieving if the received power in the $\mbf{U}$ direction exceeds $\scostb$.  This received power is:
	\begin{align*}
	\gamma(\boost,\corr)
	&= \frac{ \left( \xcostb + (1 + \boost) \corr \sqrt{\xcostb \st} + \boost \st \right)^2
		}{ P_U }.
	\end{align*}
Thus for $(\boost,\corr) \in \mc{A}(\scostb)$ the the GAVC threshold for the $\mbf{U}$ codebook can be met and $\mbf{U}$ can be decoded.   Lemma \ref{lem:good_enc} shows that for large $n$ the encoding will succeed, so the probability of error can be made as small as we like.
\end{proof}

\subsection{Examples}

Figure \ref{fig:achRplot} shows an example of the achievable rate versus $\xcostb$.   The two circles show the thresholds given by Corollary \ref{cor:cap} and the threshold for the standard Gaussian AVC with deterministic coding and average error.  The presence of the known interference $\mbf{T}$ extends the capacity region relative to the standard AVC and achieves capacity for values of $\xcostb$ that are smaller than the jammer constraint $\scostb$.  Thus far we have been unable to improve the converse for the region in which DPC does not achieve capacity; it may be that a different coding scheme exploiting the interference $\mbf{T}$ can achieve higher rates in this regime.

\begin{figure}
\centering 
\includegraphics[width=3.0in]{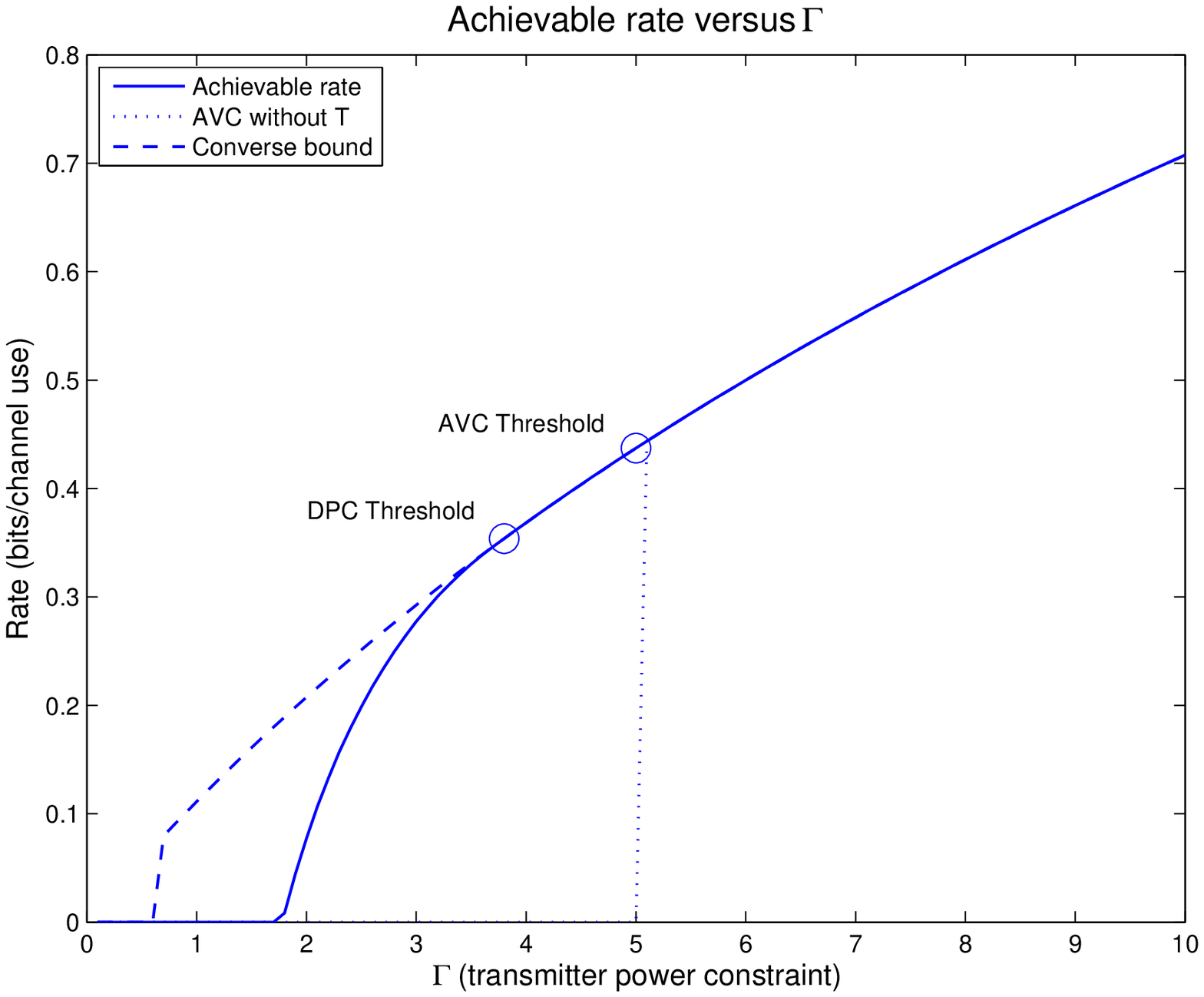}
\caption{Rates versus $\xcostb$ for $\scostb = 5$, $\st = 2$, and $\sw = 1$.  The solid line is the achievable rate and the dashed line is the outer bound.  The dotted line is the AVC capacity without the known interference signal.  The threshold for the dotted line is at $\xcostb = \scostb$ and the DPC threshold is given by \eqref{eq:st_thresh}.}
\label{fig:achRplot} 
\end{figure}

One application of dirty paper codes is in watermarking, in which an
encoder must encode a message $m$ in a given \textit{covertext}
(e.g. an image) which is modeled by an iid Gaussian sequence
$\mbf{T}$.  The encoder produces a \textit{stegotext} $\mbf{V} =
\denc(m, \mbf{T}) + \mbf{T}$ that satisfies a distortion constraint
$\norm{\mbf{V} - \mbf{T}} = \norm{\denc(m, \mbf{T})}^2 \le \xcostb$.
A limited class of attacks are \textit{additive attacks}, which take
the form of an additive signal $\mbf{s}$ that is independent of the
stegotext $\mbf{V}$ such that the receiver gets $\mbf{Y} = \mbf{s} +
\mbf{V}$.  In this model, if the encoder and decoder share common
randomness, then standard Gaussian AVC results imply that the highest
rate that can be transmitted via the stegotext is the dirty paper
coding capacity with $\mbf{s}$ equal to Gaussian noise.  We can call
this the randomized watermarking capacity.

By contrast, if there is no common randomness we can use Theorem
\ref{thm:dpc_thm} with the noise $\mbf{W}$ set to $0$ to find
achievable rates for this problem under deterministic coding; a
decoder should be able to read the watermark without sharing a secret
key with the encoder.  Because the encoder does not want to distort
the covertext by too much, an interesting regime for deterministic
watermarking is when the power $\st$ of the covertext is much higher
than the distortion limit $\xcostb$ of the encoder.  From
\eqref{eq:st_thresh} we can see that large $\st$ benefits the encoder
by increasing the effective power of the auxiliary codebook to beat
the jammer $\Lambda$.

By setting equality in \eqref{eq:st_thresh}, $\corr = 0$, and $\boost = \optboost$, we can solve for $\st$.  Let $\beta = \scostb/\xcostb$ be the ratio of the attack distortion to the watermark distortion.  A little algebra reveals:
	\begin{align*}
	\st = \xcostb \left( \frac{1}{2} \beta (5 + 4 \beta)^{1/2} 
		- \frac{1}{2} \beta
		- 1
		\right).
	\end{align*}
As we can see, the required $\st$ grows like $\beta^{3/2}$.   Therefore for a fixed watermark distortion, the cover text variance must increase like $\scostb^{3/2}$ in order to communicate at the randomized watermarking capacity.

\section{Rank-limited jammers}

In this section we study a model for multiple-input multiple-output (MIMO) Gaussian channels \cite{Telatar:99mimo} under limited jamming.  Hughes and Narayan \cite{HughesN:88vgavc} found the randomized coding capacity of an $M \times M$ MIMO AVC when the jammer has $M$ antennas as well.  A game theoretic model for this problem, carried out most fully by Baker and Chao \cite{BakerC:96game}, uses the mutual information as a payoff between one player who can choose a transmit covariance matrix and another who can choose the noise covariance matrix.  

\begin{figure}
\begin{center}
\psset{xunit=0.05cm,yunit=0.05cm,runit=0.05cm}
\begin{pspicture}(-10,0)(160,140)

\rput(75,30){
	\begin{pspicture}(0,0)(150,70)
	
		\rput(20,50){\textbf{Configuration 2}}
	
		\rput(30,20){
			\begin{pspicture}(0,0)(30,30)
			\psframe(0,0)(20,30)
			\rput(10,15){\textsc{Tx}}
			\psline(20,5)(25,5)(25,10)(22,13)(28,13)(25,10)
			\psline(20,20)(25,20)(25,25)(22,28)(28,28)(25,25)
			\end{pspicture}
		}

		\rput(120,20){
			\begin{pspicture}(0,0)(30,30)
			\psframe(10,0)(30,30)
			\rput(20,15){\textsc{Rx}}
			\psline(10,5)(5,5)(5,10)(2,13)(8,13)(5,10)
			\psline(10,20)(5,20)(5,25)(2,28)(8,28)(5,25)
			\end{pspicture}
		}
		
		\rput(65,15){
			\begin{pspicture}(0,0)(30,20)
			\psframe(0,0)(20,15)
			\rput(10,6.5){\textsc{Int}}
			\psline(20,5)(25,5)(25,10)(22,13)(28,13)(25,10)
			\end{pspicture}
			}

		\psline[linewidth=1,arrows=->](80,20)(95,30)(95,25)(105,32)
		\psline[linewidth=1,arrows=->](80,18)(95,11)(95,18)(105,15)
		
	\end{pspicture}
	}

\rput(75,90){
	\begin{pspicture}(0,0)(150,60)
	
		\rput(20,55){\textbf{Configuration 1}}
	
		\rput(30,30){
			\begin{pspicture}(0,0)(30,30)
			\psframe(0,0)(20,30)
			\rput(10,15){\textsc{Tx}}
			\psline(20,5)(25,5)(25,10)(22,13)(28,13)(25,10)
			\psline(20,20)(25,20)(25,25)(22,28)(28,28)(25,25)
			\end{pspicture}
		}

		\rput(120,30){
			\begin{pspicture}(0,0)(30,30)
			\psframe(10,0)(30,30)
			\rput(20,15){\textsc{Rx}}
			\psline(10,5)(5,5)(5,10)(2,13)(8,13)(5,10)
			\psline(10,20)(5,20)(5,25)(2,28)(8,28)(5,25)
			\end{pspicture}
		}
		
		\rput(75,60){
			\begin{pspicture}(0,0)(30,20)
			\psframe(0,0)(20,15)
			\rput(10,6.5){\textsc{Int}}
			\psline(20,5)(25,5)(25,10)(22,13)(28,13)(25,10)
			\end{pspicture}
			}

		\psline[linewidth=2,arrows=->](85,50)(95,43)(95,47)(105,40)
		\psline[linewidth=0.5,arrows=->,linestyle=dashed](85,45)(95,25)(95,35)(105,25)
		
	\end{pspicture}
	}

\psline[linewidth=3pt,linestyle=dashed](0,65)(150,65)

\end{pspicture}
\caption{Two different configurations for a system with a single-antenna interferer.  Because the location of the interferer is unknown, the subspace in which the interference lies may be unknown prior to transmission.  \label{fig:intcartoon}}
\end{center}
\end{figure}
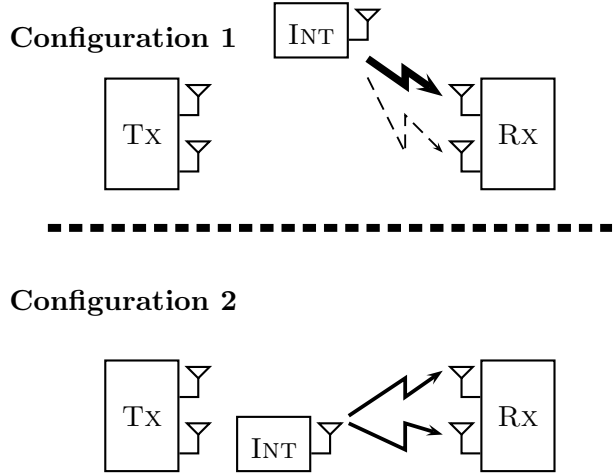

In this section we also consider fully randomized coding, but restrict the jammer to have only a single antenna.  This means that the set of noise-plus-interference covariance matrices is no longer convex and requires new analysis.  When we limit the jammer's degrees of freedom, characterizing capacity becomes more difficult.  To illustrate our problem, consider the two possible configurations shown in Figure \ref{fig:intcartoon}.  A $2 \times 2$ MIMO system is subject to unknown interference from a single antenna system.  Because the location of the interferer is not known prior to transmission, the MIMO system must choose a rate and coding scheme that will work regardless of the interferer's location.  The fact that the interferer has a single antenna means that the interference lies in an unknown one-dimensional subspace of the received signal.  We show that this limitation can be exploited to achieve rates higher than the full-rank jammer \cite{HughesN:88vgavc}.  To focus on these rank effects, we assume the encoder and decoder share common randomness.

\subsection{Channel model}

For simplicity, we will treat our MIMO channel as a vector Gaussian channel.  Over a blocklength $n$, the  channel is given by
	\begin{align}
	\mbf{Y} = \mbf{X} + \mbf{s}\mbf{g}^T + \mbf{W},
	\label{eq:mimochan}
	\end{align}
where $\mbf{X}$, $\mbf{Y}$, and $\mbf{W}$ taking values in $\mathbb{R}^{n \times M}$, $\mbf{g}$ is an arbitrary unit vector in $\mathbb{R}^M$, the interference $\mbf{s}$ is subject to the same average power constraint $\norm{\mbf{s}}^2 \le n \jpwr$, and each row of $\mbf{W}$ is i.i.d.~with distribution $\mc{N}(0,\Sw)$, where $\Sw$ is a positive definite $M \times M$ covariance matrix.  The transmitter is also subject to a sum power constraint $\norm{\mbf{X}}^2 \le n \xpwr$.  We can, without loss of generality, take the noise covariance matrix to be diagonal, so $\Sw = \diag(\noi{1}, \noi{2}, \ldots, \noi{M})$.
The interference is constrained to a rank-1 subspace, albeit an unknown one.  We must therefore design a coding scheme that works for all values of $\mbf{g}$.
We call such a channel an $(M,M,1)$ MIMO AVC.  It is easy to generalize this model to $(M_T,M_R,M_J)$ MIMO AVCs with $M_T$ transmit antennas, $M_R$ receive antennas, and $M_J$ jamming antennas.

An $(n,N)$ deterministic code with power constraint $\xpwr$ for this channel is a pair of functions $(\denc,\ddec)$, where $\denc : [N] \to \mathbb{R}^{n \times M}$ and $\ddec : \mathbb{R}^{n \times M} \to [N]$ and
	\begin{align*}
	\norm{ \denc(i) }^{2} \le n \xpwr \qquad \forall i \in [N].
	\end{align*}
An $(n,N)$ randomized code with power constraint $\xpwr$ is a random variable $(\renc,\rdec)$ taking values in the set of $(n,N)$ deterministic codes.  The maximal probability of error for a randomized code under a rank-$1$ jammer with power $\jpwr$ for the channel \eqref{eq:mimochan} is
	\begin{align*}
	\varepsilon 
	= \max_{ \mbf{s} \in \mathbb{R}^n : \norm{\mbf{s}} \le n \jpwr }
		\max_{ \mbf{g} \in \mathbb{R}^M : \norm{ \mbf{g} } = 1 }
		\max_{i \in [N]}
		\prob\left( \rdec( \renc(i) + \mbf{s} \mbf{g}^T + \mbf{W} ) \ne i \right).
	\end{align*}
A rate $R$ is achievable under randomized coding and maximal error if there exists a sequence of $(n,\lceil \exp(nR) \rceil)$ randomized codes whose maximal error goes to $0$ as $n \to \infty$.  The randomized coding capacity $C_r$ is the supremum of the achievable rates.

\subsection{Main Result}

In the case without the rank constraint on the interference, the jammer can also allocate power to all the degrees of freedom in this channel.  This channel is equivalent to a vector Gaussian AVC \cite{HughesN:88vgavc} and the capacity for general $M$ under randomized coding is known to be given by a ``mutual waterfilling'' strategy.  Both the transmitter and jammer choose diagonal covariance matrices.  The jammer chooses a covariance $\diag(\jpwr_1, \jpwr_2, \ldots, \jpwr_M)$ by waterfilling over the noise spectrum:
	\begin{align}
	\lambda^{\ast} &= \max\left\{ \lambda : \left( \lambda - \noi{m} \right)^{+} \le \jpwr \right\} \label{eq:jam_wfill1} \\
	\jpwr_m &= \left( \lambda^{\ast} - \noi{m} \right)^{+}.
	\end{align}
The transmitter then chooses a covariance $\diag(\xpwr_1, \xpwr_2, \ldots, \xpwr_M)$ based on this worst jamming strategy:
	\begin{align}
	\gamma^{\ast} &= \max\left\{ \gamma : \left( \gamma - \noi{m} - \jpwr_m \right)^{+} \le \xpwr \right\} \\
	\xpwr_m &= \left( \gamma^{\ast} - \noi{m} - \jpwr_m \right)^{+}.	
	\label{eq:tx_wfill2}
	\end{align}
Hughes and Narayan \cite{HughesN:88vgavc} showed that this allocation is a saddle point for the mutual information and is achievable for the Gaussian AVC with randomized coding.  Later, Csisz\'{a}r \cite{Csiszar:91general} showed that the capacity for deterministic codes is also given by this allocation if $\xpwr > \jpwr$.

By treating the jammer as if it has $M$ antennas, the capacity of the vector Gaussian AVC \cite{HughesN:88vgavc} is an achievable rate for the $(M,M,1)$ MIMO AVC model. 

\begin{theorem}[Full rank jammer \cite{HughesN:88vgavc}]
For the $(M,M,1)$ MIMO AVC, the following rate is achievable using randomized coding:
	\begin{align}
	R_{\mathrm{wfill}} = \sum_{m=1}^{M} \frac{1}{2} \log\left( 1 + \frac{\xpwr_m}{\jpwr_m + \noi{m}} \right),
	\label{eq:rate:wfill}
	\end{align}
where $\{\xpwr_m\}$ and $\{\jpwr_m\}$ are given by the waterfilling solutions in \eqref{eq:jam_wfill1}--\eqref{eq:tx_wfill2}.
\end{theorem}

However, the rank constraint on the jammer should admit rates higher than $R_{\mathrm{wfill}}$, since in many cases the jammer's waterfilling strategy does not satisfy its rank constraint.   By examining the arguments of Hughes and Narayan \cite{HughesN:88vgavc}, we can find an achievable rate for this channel.  If the transmitter fixes a covariance matrix $\Sx$ first, then the following rate is achievable:
	\begin{align}
	R^{\ast} = \max_{\Sx : \tr(\Sx) \le \xpwr} \min_{\mbf{g} : \norm{\mbf{g}} = 1} 
		\frac{1}{2} \log \frac{ \det(\Sx + \Sw + \jpwr \mbf{g}\mbf{g}^T) }{ \det(\Sw + \jpwr \mbf{g}\mbf{g}^T) }.
	\label{eq:rankdef_maxmin}
	\end{align}
Unfortunately, even the inner minimization is not convex in general, so standard optimization techniques are difficult to apply.  If the max-min is equal to the min-max then this expression is the capacity.  An optimistic upper bound on the capacity can be found by assuming the jammer adds $\jpwr$ to the sub-channel with the weakest noise and then using the waterfilling solution for the transmitter, which proves the following theorem.

\begin{theorem}
\label{thm:mimoupper}
Suppose $\noi{1} \le \noi{2} \le \cdots \le \noi{M}$ and let $\tau_1 = \noi{1} + \jpwr$ and $\tau_m = \noi{m}$ for $m \ge 2$.  Then for
	\begin{align*}
	\gamma^{\ast} &= \max \left\{ \gamma : \sum_{m=1}^{M} (\gamma - \tau_m)^{+} \le \Gamma \right\} \\
	\xpwr_m &= (\gamma - \tau_m)^{+},
	\end{align*}
the capacity of the $(M,M,1)$ MIMO AVC is upper bounded by
	\begin{align}
	R_{\mathrm{ub}} = \sum_{m=1}^{M} \frac{1}{2} \log\left( 1 + \frac{ \xpwr_m }{ \tau_m } \right).
	\label{eq:mimoub}
	\end{align}
\end{theorem}

Our main result is a characterization of the optimal strategies in \eqref{eq:rankdef_maxmin}.  The following theorem shows that the maximizing input covariance $\Sx$ is diagonal and the corresponding minimizing jamming strategy is to jam a single subchannel.

\begin{theorem}
\label{thm:optjamdir}
The input covariance matrix $\Sx$ maximizing the rate (\ref{eq:rankdef_maxmin}) for the $(M,M,1)$ MIMO AVC is diagonal.  Suppose $\Sx = \diag(\xpwr_1, \xpwr_2, \ldots, \xpwr_M)$.  Then the worst-case jamming direction $\mbf{g}$ is equal to $\mbf{e}_m$, where
	\begin{align}
	m = \argmax_{i \in [M]}
	\frac{ \xpwr_i / \sigma_i^2 }{ \xpwr_i + \sigma_i^2 + \Lambda }.
	\label{eq:optjamcond}
	\end{align}
\end{theorem}

As an example, in some cases the same waterfilling allocation for the transmitter can achieve rates higher than \eqref{eq:rate:wfill}.

\begin{corollary}
\label{cor:MM1rate}
Let $\noi{1} \le \noi{2} \le \cdots \le \noi{M}$, and let $\{\xpwr_m\}$ 
be given by the waterfilling allocation in \eqref{eq:jam_wfill1} and \eqref{eq:tx_wfill2}.  The 
following rate is achievable over the $(M,M,1)$ MIMO AVC:
	\begin{align}
	R = \frac{1}{2} \log\left( 1 + \frac{ \xpwr_1 }{ \noi{1} + \jpwr } \right)
		+ \sum_{m=2}^{M} \frac{1}{2}
			\log \left( 1 + \frac{\xpwr_m}{ \noi{m} } \right).
	\label{eq:rate:wfillnew}
	\end{align}
If $\jpwr \le \noi{2} - \noi{1}$ then this rate is equal to $R_{\mathrm{ub}}$ in \eqref{eq:mimoub} and is the capacity.
If $\jpwr > \noi{2} - \noi{1}$ then this rate is larger than $R_{\mathrm{wfill}}$ in \eqref{eq:rate:wfill}.
\end{corollary}

\subsection{Analysis}

We begin with a simple technical lemma whose proof we include for completeness.

\begin{lemma}[Matrix Determinant Lemma]
\label{lem:mdl}
Let $A$ be an $M \times M$ positive definite matrix and $U$ and $V$ be two $M \times k$ matrices.  Then
	\[
	\det(A + U V^H) = \det(A) \det(I_k + V^H A^{-1} U).
	\]
\end{lemma}

\begin{proof}
First,
	\begin{align*}
	\left[
		\begin{array}{cc}
		A & -U \\
		V^{H} & I
		\end{array}
	\right]
	&= 
	\left[
		\begin{array}{cc}
		A & 0 \\
		V^{H} & I
		\end{array}
	\right]
	\cdot
	\left[
		\begin{array}{cc}
		I & - A^{-1} U \\
		0 & I + V^H A^{-1} U
		\end{array}
	\right].
	\end{align*}
Taking determinants on both sides yields the result.
\end{proof}

\begin{proof}[Proof of Theorem \ref{thm:optjamdir}]
We prove the second part of the theorem first.  Let
	\begin{align}
	L(\Sx,\mbf{g}) = \frac{1}{2} \log \frac{\det(\Sx + \Sw + \Lambda \mbf{g} \mbf{g}^T) }{ \det( \Sw + \Lambda \mbf{g} \mbf{g}^T ) }.
	\label{eq:mimorate}
	\end{align}
So we can use Lemma \ref{lem:mdl} to expand this:
	\[
	\frac{1}{2} \log \frac{ 
		\det(\Sx + \Sw) ( 1 + \Lambda \mbf{g}^T (\Sx + \Sw)^{-1} \mbf{g} ) 
		}{
		\det(\Sw) (1 + \Lambda \mbf{g}^T \Sw^{-1} \mbf{g} )
		}.
	\]
Thus minimizing over $\mbf{g}$ reduces to minimizing
	\[
	J(\mbf{g}) = \log \frac{ 1 + \Lambda \mbf{g}^T (\Sx + \Sw)^{-1} \mbf{g} }{ 1 + \Lambda \mbf{g}^T \Sw^{-1} \mbf{g} }.
	\]

Let $\mbf{g} = (\sqrt{1 - g_M^2} \mbf{h}, g_M)^T$, where $\mbf{h} \in \mathbb{R}^{M-1}$ is a unit vector.  Taking the gradient of $J(\mbf{g})$ we see:
	\begin{align*}
	\nabla J(\mbf{g}) 
	&= \frac{ \Lambda }{ 1 + \Lambda \mbf{g}^T (\Sx + \Sw)^{-1} \mbf{g} }
		(\Sx + \Sw)^{-1} \mbf{g} \\
		& \hspace{0.5in}
		- 
	  \frac{ \Lambda }{ 1 + \Lambda \mbf{g}^T \Sw^{-1} \mbf{g} }
	  	\Sw^{-1} \mbf{g}.
	\end{align*}	
We can write the function $J(\cdot)$ with respect to $g_M^2$:
	\begin{align*}
	J(g_M^2) &= \log \frac{1 + \Lambda \left( \frac{g_M^2}{\xpwr_M + \sigma_M^2} + \sum_{m=1}^{M-1} (1 - g_M^2) \frac{ h_m^2 }{\xpwr_m + \sigma_m^2 } \right) 
		}{
		1 + \Lambda \left( \frac{ g_M^2}{\sigma_M^2 } + \sum_{m=1}^{M-1} (1 - g_M^2) \frac{ h_m^2 }{\sigma_m^2} \right)
		}.
	\end{align*}
Let
	\begin{align*}
	\alpha_1 & = \sum_{m=1}^{M-1} \frac{ h_m^2 }{\xpwr_m + \sigma_m^2 } \\
	\alpha_2 &= \sum_{m=1}^{M-1} \frac{ h_m^2 }{ \sigma_m^2 }.
	\end{align*}	
Taking a derivative with respect to $g_M^2$:
	\begin{align*}
	\frac{\partial}{\partial g_M^2} J(g_M^2)
	&= \frac{ \Lambda \left( \frac{1}{\xpwr_M + \sigma_M^2} - \alpha_1 \right) 
		}{
		1 + \Lambda \left( \frac{g_M^2}{\xpwr_M + \sigma_M^2} + (1 - g_M^2) \alpha_1 \right) 
		} \\
		&\hspace{0.5in}
		-
		\frac{ \Lambda \left( \frac{1}{\sigma_M^2} - \alpha_2 \right) 
		}{
		1 + \Lambda \left( \frac{g_M^2}{\sigma_M^2} + (1 - g_M^2) \alpha_2 \right) 
		} \\
	&= 
		\frac{
		\Lambda \left( \frac{1}{\xpwr_M + \sigma_M^2} - \alpha_1 \right)
		}{
		1 + \Lambda \alpha_1 + 
		\Lambda g_M^2 \left( \frac{1}{\xpwr_M + \sigma_M^2} - \alpha_1 \right)
		} \\
		&\hspace{0.5in}
		-
		\frac{
		\Lambda \left( \frac{1}{\sigma_M^2} - \alpha_2 \right) 
		}{
		1 + \Lambda \alpha_2 + 
		\Lambda g_M^2  \left( \frac{1}{\sigma_M^2} - \alpha_2 \right) 
		}.
	\end{align*}
The derivative is positive if
	\begin{align*}
	\frac{
		\left( \frac{1}{\xpwr_M + \sigma_M^2} - \alpha_1 \right)
		}{
		1 + \Lambda \alpha_1 + 
		\Lambda g_M^2 \left( \frac{1}{\xpwr_M + \sigma_M^2} - \alpha_1 \right)
		}
	& \\
	& \hspace{-0.5in} 
	>
	\frac{
		\left( \frac{1}{\sigma_M^2} - \alpha_2 \right) 
		}{
		1 + \Lambda \alpha_2 + 
		\Lambda g_M^2  \left( \frac{1}{\sigma_M^2} - \alpha_2 \right) 
		},
	\end{align*}
or
	\begin{align*}
	\left( \frac{1}{\xpwr_M + \sigma_M^2} - \alpha_1 \right)
		(1 + \Lambda \alpha_1 ) 
	&>
	\left( \frac{1}{\sigma_M^2} - \alpha_2 \right) 
		(1 + \Lambda \alpha_2 ).
	\end{align*}
Note that this condition is independent of of $g_M$, so for any $\mbf{h}$, the optimal value of
$g_M = 1$ or $g_M = 0$.  In the first case, the jammer's optimal strategy is to choose $\mbf{g} = \mbf{e}_M$.  In the second case, $(\sqrt{1 - g_M^2} \mbf{h},g_M)$ yields a higher rate
than $(\mbf{h},0)$ so repeating the argument on $\mbf{h}$ shows that $\mbf{g} = \mbf{e}_m$ for some $m \in [M-1]$.  Therefore the optimal jammer strategy is to pick $\mbf{g}$ equal to an elementary vector.

If $\mbf{g} = \mbf{e}_m$ is optimal for the jammer, then for any $i \ne m$,
	\begin{align*}
	\frac{ \det( \Sx + \Sw + \Lambda \mbf{e}_m \mbf{e}_m^T )}{ \det( \Sw + \Lambda \mbf{e}_m \mbf{e}_m^T ) }
	< \frac{ \det( \Sx + \Sw + \Lambda \mbf{e}_i \mbf{e}_i^T )}{ \det( \Sw + \Lambda \mbf{e}_i \mbf{e}_i^T ) },
	\end{align*}	
or
	\begin{align*}
	\frac{ (\xpwr_m + \sigma_m^2 + \Lambda) ( \xpwr_i + \sigma_i^2 ) 
		}{ (\sigma_m^2 + \Lambda) \sigma_i^2 }
	&< \frac{ (\xpwr_i + \sigma_i^2 + \Lambda) ( \xpwr_m + \sigma_m^2 ) 
		}{ (\sigma_i^2 + \Lambda) \sigma_m^2 }.
	\end{align*}
Some algebra reveals that
	\begin{align*}
	(\xpwr_m + \sigma_m^2 + \Lambda) ( \xpwr_i + \sigma_i^2 ) (\sigma_i^2 + \Lambda) \sigma_m^2
	& \\
	&\hspace{-1.5in}
	< 
	(\xpwr_i + \sigma_i^2 + \Lambda) ( \xpwr_m + \sigma_m^2 ) (\sigma_m^2 + \Lambda) \sigma_i^2,
	\end{align*}
from which it follows that
	\begin{align*}
	(\xpwr_m + \sigma_m^2 + \Lambda) ( 
		\xpwr_i \sigma_i^2 \sigma_m^2 
		+ \xpwr_i \sigma_m^2 \Lambda 
		+ \sigma_i^2 \sigma_m^2 \Lambda
		+ \sigma_i^4 \sigma_m^2 
		)
	&  \\
	&\hspace{-3.1in}
	< 
	(\xpwr_i + \sigma_i^2 + \Lambda) ( 
		\xpwr_m \sigma_i^2 \sigma_m^2 
		+ \xpwr_m \sigma_i^2 \Lambda 
		+ \sigma_i^2 \sigma_m^2 \Lambda
		+ \sigma_i^2 \sigma_m^4 
		),
	\end{align*}
and finally that
	\begin{align*}
		\xpwr_i \xpwr_m \sigma_m^2 
		+ \xpwr_i \sigma_m^4  
		+ \xpwr_i \sigma_m^2 \Lambda
	&
	< 
		\xpwr_i \xpwr_m \sigma_i^2 
		+ \xpwr_m \sigma_i^4  
		+ \xpwr_m \sigma_i^2 \Lambda.
	\end{align*}
So for a given $S_x$, the optimal jamming direction is $\mbf{e}_m$ if
	\begin{align}
	\frac{ \xpwr_m/\sigma_m^2 }{ \xpwr_m + \sigma_m^2 + \Lambda}
	&>
	\frac{ \xpwr_i / \sigma_i^2 }{ \xpwr_i + \sigma_i^2 + \Lambda }
	\end{align}
for all $i \ne m$.

Now consider $L(\Sx,\mbf{g})$ in \eqref{eq:mimorate} and suppose $\Sx$ is arbitrary.  Let $\Sx^{\circ} = \diag( \{ (\Sx)_{mm} \})$ be diagonal matrix containing the diagonal of $\Sx$.  Then by Hadamard's inequality, for any $m$,
	\begin{align*}
	L(\Sx^{\circ},\mbf{e}_m) - L(\Sx,\mbf{e}_m) 
	&= \frac{1}{2} \log \frac{ \det( \Sx^{\circ} + \Sw + \jpwr \mbf{e}_m \mbf{e}_m^T )
		}{
		\det( \Sx + \Sw + \jpwr \mbf{e}_m \mbf{e}_m^T )
		} \\
	&\ge 0.
	\end{align*}
Let $m^{\circ} = \argmin_{m} L(\Sx^{\circ},\mbf{e}_m)$.  Then
	\begin{align*}
	\min_{\mbf{g}} L(\Sx,\mbf{g}) 
	&\le L(\Sx, \mbf{e}_{m^{\circ}}) \\
	&\le L(\Sx^{\circ},\mbf{e}_{m^{\circ}}) \\
	&= \min_{\mbf{g}} L(\Sx^{\circ},\mbf{g}),
	\end{align*}
so the transmitter can always increase the rate by choosing $\Sx$ to be diagonal.
\end{proof}

\begin{proof}[Proof of Corollary \ref{cor:MM1rate}]
Let $\xpwr_1 \ge \xpwr_2 \ge \cdots \ge \xpwr_M$ be waterfilling power allocation, and let $\gamma = \sigma_1^2 + \xpwr_1$ be the ``water level.''  For this allocation, it is clear that
	\begin{align}
	\frac{ \xpwr_1/\noi{1} }{ \xpwr_1 + \noi{1} + \Lambda}
	&>
	\frac{ \xpwr_i / \noi{i} }{ \xpwr_i + \noi{i} + \Lambda }.
	\end{align}
for all $i \ne 1$, so by Theorem \ref{thm:optjamdir}, the worst-case direction $\mbf{g}$ for the jammer is $\mbf{g} = \mbf{e}_1$ and \eqref{eq:rate:wfillnew} is achievable.  If $\jpwr \le \noi{2} - \noi{1}$ then this is the optimal strategy for the full-rank jammer so \eqref{eq:rate:wfillnew} is the capacity.
If $\jpwr > \noi{2} - \noi{1}$ then comparing this to $R_{\mathrm{wfill}}$ we see that $R > R_{\mathrm{wfill}}$.
\end{proof}

\subsection{The $(2,2,1)$ MIMO AVC}

For a given diagonal covariance matrix, Theorem \ref{thm:optjamdir} shows that the jammer's optimal strategy is always to jam one of the subchannels.  The set of covariance matrices for which the optimal jamming direction is $\mbf{e}_m$ is given by \eqref{eq:optjamcond}.  Unfortunately, maximizing the rate subject to the conditions in \eqref{eq:optjamcond} does not lead to a clean solution.  In the $(2,2,1)$ MIMO AVC we can carry out the calculation explicitly.

\begin{theorem}
\label{thm:optmimo}
Let $\beta$ be the value of $\xpwr_1$ for which the terms in the maximization \eqref{eq:optjamcond} are equal:
	\[
	\frac{\beta/\noi{1}}{ \beta + \noi{1} + \jpwr } = \frac{ (\xpwr - \beta)/\noi{2} }{ \xpwr - \beta + \noi{2} + \jpwr},
	\]
let
	\[
	\gamma = \frac{1}{2} (\Gamma - (\noi{1} + \jpwr - \noi{2})),
	\]
and let
	\[
	R(\alpha) = \frac{1}{2} \log\left( 1 + \frac{\alpha}{\jpwr + \noi{1}} \right)
		+ \frac{1}{2} \log\left( 1 + \frac{\xpwr - \alpha}{ \noi{2} } \right)
	\]
Then for the $(2,2,1)$ MIMO AVC,
	\begin{itemize}
	\item if $\noi{1} + \jpwr \le \noi{2}$ then $R(\gamma)$ is the capacity,
	\item if $\noi{1} + \jpwr > \noi{2}$,  $\xpwr > \noi{1} + \jpwr - \noi{2}$ and $\gamma > \beta$ then $R(\gamma)$ is the capacity,
	\item if $\noi{1} + \jpwr > \noi{2}$, and $\xpwr \le \noi{1} + \jpwr - \noi{2}$ or $\gamma < \beta$ then $R(\beta)$ is achievable.
	\end{itemize}
\end{theorem}

\begin{proof}
Note that the optimal jamming strategy is $\mbf{g} = \mbf{e}_1$ for $\xpwr_1 > \beta$ and $\mbf{g} = \mbf{e}_2$ for $\xpwr_1 < \beta$.  Suppose first that $\noi{1} + \jpwr \le \noi{2}$.  In this case, the waterfilling power allocation for noise spectrum $(\noi{1} + \jpwr, \noi{2})$ is such that $\xpwr_1 > \beta$, so the capacity is given by \eqref{eq:rate:wfill}, coinciding with Theorem \ref{thm:mimoupper}.
	
Now suppose that $\noi{1} + \jpwr > \noi{2}$.  If $\xpwr > \noi{1} + \jpwr - \noi{2}$, then we consider two sub-cases.  If the waterfilling power allocation $(\gamma, \Gamma - \gamma)$ for noise spectrum $(\noi{1} + \jpwr, \noi{2})$ satisfies $\gamma > \beta$, then that power allocation is optimal and the rate is given by $R(\gamma)$.  However, if $\gamma \le \beta$, then setting $\xpwr_1 = \beta$ is optimal.  For $\xpwr_1 < \beta$ the optimal jammer strategy is $\mbf{e}_2$, but the achievable rate is monotonically increasing in $\xpwr_1$.  For $\xpwr_1 > \beta$ the optimal jamming strategy is $\mbf{e}_2$ and the rate is monotonically decreasing in $\xpwr_1$.  Hence $\xpwr_1 = \beta$ is optimal.
\end{proof}

From this result we can see the difficulty in \eqref{eq:rankdef_maxmin}. 
Suppose $\xpwr = 6$, $\noi{1} = 3$, $\noi{2} = 1$, and $\jpwr = 4$.  If $\mbf{g} = (0,1)^T$, then the waterfilling solution for a channel with Gaussian noise of covariance $\Sw + \jpwr \mbf{g} \mbf{g}^T = \diag(3, 5)$ is to choose $\Sx = \diag(4,2)$.  However, for this choice of $\Sx$ the optimal $\mbf{g} = (1,0)^T$.  What this shows is that the max and min in \eqref{eq:rankdef_maxmin} cannot be reversed.  In two cases the waterfilling allocation is optimal, but in the other cases the transmitter chooses its power allocation and rate such that the jammer can jam either subchannel.  The transmitter is forced to choose a power allocation different from Theorem \ref{thm:mimoupper} even when the jammer has a rank constraint.

Figure \ref{fig:ex221} shows $R_{\mathrm{wfill}}$ and the rate given by Theorem \ref{thm:optmimo} as a function of the interference power $\jpwr$.  The curves are equal until the point where $\jpwr = \noi{2} - \noi{1}$.  For such values, the jammer can realize the waterfilling strategy.  However, for larger values of $\jpwr$, the rank constraint on the jammer prohibits waterfilling across multiple channels.  Under a diagonal input covariance the optimal strategy is to allocate all of $\jpwr$ to a single channel.  For large interference powers, the rank constraint allows the transmitter and receiver to communicate at rates strictly higher than $R_{\mathrm{wfill}}$.

\begin{figure}
\centering 
\includegraphics[width=3.0in]{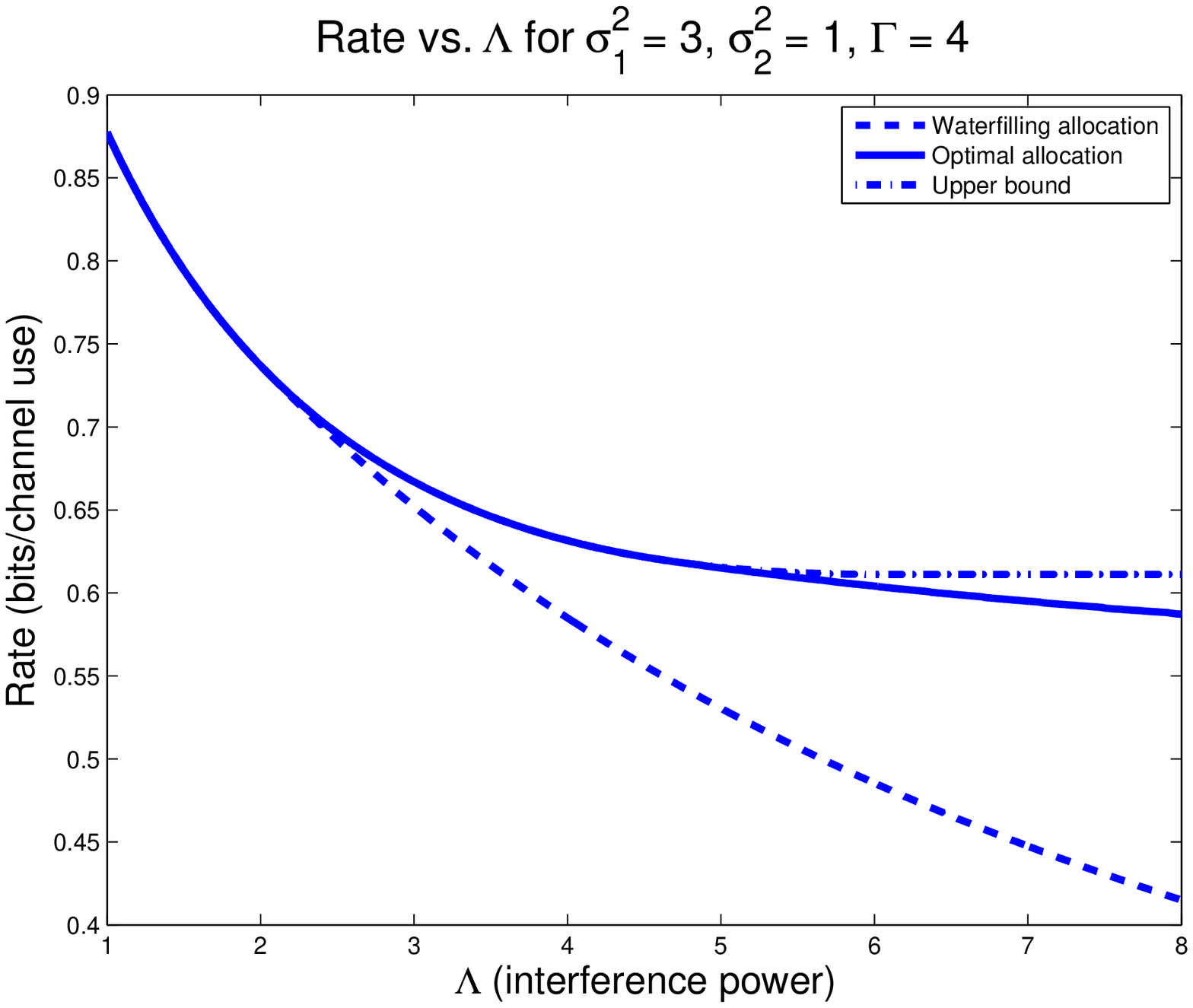}
\caption{An example of the achievable rates versus the interference power $\jpwr$ for $\noi{1} = 1$, $\noi{2} = 3$, $\xpwr = 4$.}
\label{fig:ex221} 
\end{figure}

We can also examine the asymptotic behavior of the capacity as $\jpwr \to \infty$.  The optimal jamming strategy is still to jam the less noisy channel, so the noise-plus-interference spectrum becomes more and more unbalanced.  Clearly the subchannel with noise $\noi{1} + \jpwr$ contributes no rate to the capacity in the limit.  However, any power in the first subchannel will still contribute.  As $\jpwr \to \infty$ the limiting behavior is given by the threshold condition in \eqref{eq:optjamcond}. The optimal power allocation is given by
	\begin{align*}
	\frac{\xpwr_1}{\noi{1}} = \frac{\xpwr_2}{\noi{2}}.
	\end{align*}

\begin{corollary}
For the $(2,2,1)$ MIMO AVC with $\noi{1} < \noi{2}$, the following rate is achievable in the limit as $\jpwr \to \infty$:
	\begin{align*}
	R = \frac{1}{2} \log\left(1 + \frac{\xpwr}{\noi{1} + \noi{2}} \right).
	\end{align*}
\end{corollary}

Note that this rate is less than $R_{\mathrm{ub}}= \frac{1}{2} \log \left( 1 + \frac{\xpwr}{\noi{2}} \right)$ in \eqref{eq:mimoub}.  We conjecture that this loss with respect to the optimal strategy knowing the jammer's strategy is inherent due to the adversarial nature of the channel.

A different regime is when $\xpwr$ and $\jpwr$ go to $\infty$ while keeping the ratio $\rho = \xpwr/\jpwr$ fixed.  
 
\begin{corollary}
For the $(2,2,1)$ MIMO AVC with $\noi{1} < \noi{2}$, if $\xpwr,\jpwr \to \infty$ with fixed $\rho = \xpwr/\jpwr$ then the achievable rate scales according to
	\begin{align*}
	R(\rho,\xpwr) = O(\log \xpwr) 
		+ \frac{1}{2} \log\left(1 + \frac{\rho}{2} \right).
	\end{align*}
\end{corollary}

The results here are a first step towards understanding the effect of rank-limited uncertainty in interference for MIMO systems.  There are a number of interesting open questions for future work.  Firstly, finding the optimal rate in Theorem \ref{thm:optjamdir} requires optimizing over different sets of power allocations corresponding to different optimal jamming strategies.  We conjecture that the optimal rate always corresponds to the jammer setting $\mbf{g} = \mbf{e}_1$, where $\sigma_1^2$ is the smallest noise variance.  Secondly, showing that this rate is indeed the capacity of the MIMO AVC requires different techniques than the full rank case, where the jammer and transmitter strategies formed a saddle point.  Finally, extending these results to more general numbers of transmit, receive, and jamming antennas would be very interesting and may shed some light into other problems in rank-limited optimization.

\section{Conclusion}

In this paper we investigated some variations on the basic Gaussian AVC model to illustrate different aspects of coding for worst-case interference.  One way of interpreting these results is as intermediate stages between worst-case and average-case analysis.  The worst-case interference in the GAVC can depend on the codebook of the transmitted message.  However, with additional resources, this worst-case behavior can be relaxed to attain rates closer to the average-case behavior.  We demonstrated that a very small amount of common randomness is sufficient to achieve the randomized coding capacity of the GAVC, that a known interference signal can help mask the codeword and allow reliable communication with moderate interference, and that extra degrees of freedom can overcome even worst-case interference.  

There are still several open questions that remain.  Is the $O(\log n)$ bits of common randomness necessary to achieve the randomized coding capacity?  Finding lower bounds on the amount of randomness can quantify how close the worst case is to the average case.  Is the dirty-paper coding scheme optimal?  We conjecture that it is, in the sense that the encoder cannot exploit the known interference beyond the geometric approach that we describe.  For the general MIMO AVC with rank-limited jammer is there a simple characterization of the optimal transmitter and jammer power allocations?  We conjecture that the worst-case interference jams the strongest sub-channel.

Our results show that the analysis of coding schemes robust to unknown interference is highly dependent on the resources available to the encoder and decoder.  The arguments here are geometric in nature, and it may be interesting to pursue the relationship between worst-case and average-case structures for other high-dimensional problems.

\bibliographystyle{IEEEtran}
\bibliography{gavc}

\end{document}